\def\notes{1}
\def\compact{1}
\documentclass{article}
\usepackage{amsmath}
\usepackage{amssymb}
\usepackage{algorithm}
\usepackage{algorithmicx}
\usepackage[noend]{algpseudocode}
\usepackage{amsthm}
\usepackage{macros}
\usepackage{fullpage}
\usepackage{xcolor}
\usepackage{thmtools,enumitem}
\ifnum\compact=1 %
\setitemize{noitemsep,topsep=1pt,parsep=0pt,partopsep=0pt}
\if

\DeclareMathOperator{\Exp}{\mathbb{E}}

\title{Local Lipschitz Filters for Bounded-Range Functions
with Applications to Arbitrary Real-Valued Functions}
\date{}
\author{Jane Lange\thanks{MIT, {\tt jlange@mit.edu}. Supported in part by NSF Graduate Research Fellowship under Grant No. 2141064, NSF Awards CCF-2006664, DMS-2022448, CCF-2310818 and Microsoft Research. } 
\and
Ephraim Linder\thanks{\texttt{ejlinder@bu.edu}. Boston University, Department of Computer Science.}
\and
Sofya Raskhodnikova\thanks{\texttt{sofya@bu.edu}. Boston University, Department of Computer Science.}
\and
Arsen Vasilyan\thanks{\texttt{vasilyan@mit.edu}. Supported in part by NSF awards CCF-2006664, DMS-2022448, CCF-1565235, CCF-1955217,\\ CCF-2310818, Big George Fellowship and Fintech@CSAIL.}
}

\begin{document}
\maketitle

\begin{abstract}
We study local filters for the Lipschitz property of real-valued functions $f: V \to [0,r]$, where the Lipschitz property is defined with respect to an arbitrary undirected graph $G=(V,E)$. We give nearly optimal local Lipschitz filters both with respect to $\ell_1$-distance and $\ell_0$-distance. Previous work only considered unbounded-range functions over $[n]^d$. Jha and Raskhodnikova (SICOMP `13) gave an algorithm for such functions with lookup complexity exponential in $d$, which Awasthi et al.\ (ACM Trans.\ Comput.\ Theory) showed was necessary in this setting. We demonstrate that important applications of local Lipschitz filters can be accomplished with filters for functions whose range is bounded in $[0,r]$. For functions $f: [n]^d\to [0,r]$, we achieve running time $(d^r\log n)^{O(\log r)}$ for the $\ell_1$-respecting filter and $d^{O(r)}\polylog n$ for the $\ell_0$-respecting filter, thus circumventing the lower bound. Our local filters provide a novel Lipschitz extension that can be implemented locally. Furthermore, we show that our algorithms are nearly optimal in terms of the dependence on $r$ for the domain $\{0,1\}^d$, an important special case of the domain $[n]^d$. In addition, our lower bound resolves an open question of Awasthi et al., removing one of the conditions necessary for their lower bound for general range. We prove our lower bound via a reduction from distribution-free Lipschitz testing and  a new technique for proving hardness for {\em adaptive} algorithms. 

Finally, we provide two applications of our local filters to real-valued functions, with no restrictions on the range. In the first application, we use them in conjunction with the Laplace mechanism for differential privacy and noisy binary search to provide mechanisms for privately releasing outputs of black-box functions, even in the presence of malicious clients. In particular, our differentially private mechanism for arbitrary real-valued functions runs in time $2^{\polylog\min(r,nd)}$ and, for honest clients, has accuracy comparable to the Laplace mechanism for Lipschitz functions, up to a factor of $O(\log\min(r,nd))$. In the second application, we use our local filters to obtain the first nontrivial tolerant tester for the Lipschitz property. Our tester works for functions of the form $f:\{0,1\}^d\to\mathbb{R}$, makes $2^{\widetilde{O}(\sqrt{d})}$ queries, and has tolerance ratio $2.01$. Our applications demonstrate that local filters for bounded-range functions can be applied to construct efficient algorithms for arbitrary real-valued functions.

\end{abstract}

\thispagestyle{empty}

\clearpage
\pagenumbering{arabic}

\clearpage
\pagenumbering{arabic}

\section{Introduction}\label{sec:intro}

We study local Lipschitz filters for real-valued functions.
Local Lipschitz filters were first investigated by Jha and Raskhodnikova \cite{JhaR13} who motivated their 
research
by an application in private data analysis. Intuitively, a local filter for some property of functions (in our case, the Lipschitz property) is a randomized algorithm that gets oracle access to  a function $f$ and locally reconstructs the desired property in the following sense: it provides query access to a related function $g$ that is guaranteed to have the property (in our case, guaranteed to be Lipschitz). The implicit output function $g$ may depend on the internal randomness of the algorithm, but not on the order of queries. When the input function $f$ has the desired property, then $g=f$.  If, in addition, the distance between $f$ and $g$ is relatively small compared to the distance from $f$ to the nearest function with the desired property, the filter is called {\em distance-respecting.} The goal in the design of local filters is to minimize the running time and the number of {\em lookups}\footnote{Oracle calls {\em by} the filter are called {\em lookups} to distinguish them from the queries made {\em to} the filter.}, i.e., oracle calls to the input function $f$. 

The computational task performed by local filters is called {\em local reconstruction}. It was introduced
 by Saks and Seshadhri \cite{SaksS10} and is 
one of the fundamental tasks studied in the area of local computation~\cite{RubinfeldTVX11,alonLocalComputation} and sublinear-time algorithms. It has been studied for properties of functions, including monotonicity~\cite{SaksS10,bgjjrw2012,AJMR15,LangeRV23,LangeV23} and the Lipschitz property~\cite{JhaR13,AJMR15}, as well as for properties of graphs~\cite{CampagnaGR13}.

Local filters are useful in applications where some algorithm $\cA$ computing on a large dataset requires that its input satisfy a certain property. For example, in the application to privacy, which we will discuss in detail later, correctness of algorithm $\cA$ is contingent upon the input function $f$ being Lipschitz. In such applications, rather than directly relying on the oracle for $f$, algorithm $\cA$ can access its input via a local filter that guarantees that the output will satisfy the desired property, modifying $f$ on the fly if necessary. Local filters can also be used in distributed settings, where multiple processes access different parts of the input, as well as in other applications described in previous work \cite{SaksS10,bgjjrw2012,JhaR13,AJMR15}. Local reconstruction is also naturally related to other computational tasks and, for example, has been recently used to improve learning algorithms for monotone functions \cite{LangeRV23,LangeV23}.

\subsection{Our Contributions}\label{sec:contributions}

We demonstrate that important applications of local Lipschitz filters can be accomplished with computational objects that are much weaker than local Lipschitz filters for general functions: it suffices to construct local Lipschitz filters for bounded-range functions. This holds even for applications that deal with arbitrary real-valued functions, with no a priori bound on the range. We achieve efficient local Lipschitz filters for bounded-range functions, circumventing the existing lower bounds that are exponential in the dimension and enabling applications to real-valued functions, with no restriction on the range. 

\subsubsection{Local Lipschitz Filters}
Motivated by the applications,
we consider functions over $[n]^d$, where $[n]$ is a shorthand for $\{1,\dots,n\}$.
A function $f:[n]^d \rightarrow \R$ is called {\em $c$-Lipschitz} if increasing or decreasing any coordinate by one can only change the function value by $c$.
The parameter $c$ is called the {\em Lipschitz constant} of $f$. A 1-Lipschitz function is simply referred to as Lipschitz\footnote{In this work, we focus on reconstruction to Lipschitz functions.  All our results extend to the class of $c$-Lipschitz functions via scaling all function values by a factor of $c$.}. Intuitively, changing the argument to the Lipschitz function by a small amount does not significantly change the value of the function.

In previous work, only unbounded-range functions were considered in the context of Lipschitz reconstruction. Jha and Raskhodnikova~\cite{JhaR13} obtained a deterministic %
local filter %
that runs in time $O((\log n +1)^d)$. %
This direction of research was halted by a strong lower bound obtained by Awasthi et al.~\cite{AwasthiJMR16}. They showed that every local Lipschitz filter, even with significant additive error, needs exponential in the dimension $d$ number of lookups.

We demonstrate that important applications of local Lipschitz filters can be accomplished with filters for functions whose range is bounded in $[0,r]$. By focusing on this class of functions, %
we circumvent the lower bound from~\cite{AwasthiJMR16} and achieve running time polynomial in $d$ for constant $r$. Moreover, our filters satisfy additional accuracy guarantees compared to the filter in~\cite{JhaR13}, which is only required to (1) give access to a Lipschitz function~$g$; (2) ensure that $g=f$ if the input function $f$ is Lipschitz. 
Our filters achieve an additional feature of being \emph{distance-respecting}, i.e., they ensure that $g$ is close to $f$.
We provide this feature w.r.t.\ both $\ell_1$ and $\ell_0$-distance.
The $\ell_1$-distance between functions $f$ and $g$ is defined by $\|f - g\|_1 = \E_x [|f(x) - g(x)|]$, and the $\ell_0$-distance is defined by 
$\|f- g\|_0= \Pr_x[f(x) \ne g(x)]$, where the expectation and the probability are taken over a uniformly distributed point in the domain. The distance of $f$ to Lipschitzness is defined as the minimum over all Lipschitz functions $g$ of the distance from $f$ to $g$, and can be considered with respect to both norms. Our filters are {\em distance-respecting} in the following sense: the distance between the input function $f$ and the output function $g$ is at most twice the distance of $f$ to the Lipschitz property (w.h.p.). %

Our algorithms work for functions over general graphs. To facilitate
comparison with prior work
and our lower bound, we state their guarantees only for the $[n]^d$ domain.
\ifnum\compact=1
(It 
\else
(This domain
\fi
can be represented by the $d$-dimensional hypergrid graph $\hypergrid$.)  Our first local Lipschitz filter is distance-respecting 
\ifnum\compact=1
w.r.t.\
\else
with respect to 
\fi
the  $\ell_1$-distance. 
\newtheorem*{thm 1}{Theorem \ref{thm: ell_1 LCA} ($\ell_1$-filter, informal)}
\begin{thm 1}
There is an algorithm $\mathcal{A}%
$ 
that, given lookup access to a function $f:[n]^d \to [0,r]$, a query $x\in [n]^d$, and a random seed $\rho\in\{0,1\}^\star$, %
has the following properties: 
\begin{itemize}[left=5pt .. \parindent]
   \item \textbf{Efficiency: } $\mathcal{A}$ has lookup and time complexity $(d^r \cdot \polylog n)^{O(\log r)}$ per query.
   \item \textbf{Consistency: }
   With probability at least $1-n^{-d}$ over the choice of $\rho$, algorithm $\mathcal{A}$ provides query access to a $1.01$-Lipschitz function $g_{\rho}$ with $\|g_\rho - f\|_1$ at most twice the $\ell_1$-distance from $f$ to Lipschitzness. %

   If $f$ is Lipschitz, then the filter outputs $f(x)$ for all queries $x$ and random seeds.
\end{itemize}
\end{thm 1}

Our second local Lipschitz filter is distance-respecting w.r.t.\ $\ell_0.$ Unlike the filter in~\Cref{thm: ell_1 LCA}, the $\ell_0$-respecting filter provides access to a 1-Lipschitz function.

\newtheorem*{thm 2}{Theorem \ref{thm: ell_0 LCA} ($\ell_0$-filter, informal)}
\begin{thm 2}
There is an algorithm $\mathcal{A}$ that, given lookup access to a function $f:[n]^d \to [0,r]$, a query $x\in [n]^d$, and a random seed $\rho\in\{0,1\}^\star$,
has the following properties: 
\begin{itemize}[left=5pt .. \parindent]
    \item \textbf{Efficiency: } $\mathcal{A}$ has lookup and time complexity
    $d^{O(r)}\polylog(n)$.
    \item \textbf{Consistency: }
    With probability at least $1-n^{-d}$ over the choice of $\rho$, algorithm $\mathcal{A}$ provides query access to a 1-Lipschitz function $g_{\rho}$ 
    with $\|g_\rho - f\|_0$ at most twice the $\ell_0$-distance from $f$ to Lipschitzness.
    
    If $f$ is Lipschitz, then the filter outputs $f(x)$ for all queries $x$ and random seeds.
    
\end{itemize}
\end{thm 2}

An important special case of the hypergrid domain %
is %
the {\em hypercube}, denoted $\hypercube$. (It corresponds to the case $n=2$, but its vertex set is usually represented by $\{0,1\}^d$ instead of $[2]^d$.)
Prior to our work, no local Lipschitz filter for the hypercube domain
could avoid lookups on the entire domain (in the worst case). Our algorithms do that for the case when $r \le d/\log^2 d$.
Moreover, we show that our filters are nearly optimal in terms of their dependence on $r$ for this domain.
Our next theorem shows that the running time of $d^{\widetilde{\Omega}(r)}$ is unavoidable in Theorems~\ref{thm: ell_1 LCA} and~\ref{thm: ell_0 LCA}; thus, our filters are nearly optimal. Moreover, even local Lipschitz filters that are not
distance-respecting
(as in \cite{JhaR13})  
must still run in time
$d^{\widetilde{\Omega}(r)}$.

\newtheorem*{thm 3}{Theorem \ref{thm: lipschitz filter lower bound} (Filter lower bound, informal)}
\begin{thm 3}
    Let $\cA(x,\rho)$ be an algorithm that, given lookup access to a function $f:\zo^d\to[0,r]$, a query $x\in\zo^d$, and a random seed $\rho\in\zo^*$, has the following property:
    \begin{itemize}[left=5pt .. \parindent]
        \item\textbf{Weak Consistency:} With probability at least $\frac34$ over the choice of $\rho$, the algorithm $\cA(\cdot,\rho)$ provides query access to a 1-Lipschitz function $g_\rho$ such that whenever $f$ is 1-Lipschitz $g_\rho=f$. 
    \end{itemize}
    Then, for all integers $r\geq 4$ and $d \geq \Omega(r)$, there exists a function $f:\zo^d\to[0,r]$ for which the  lookup complexity of $\cA$ is $(\frac dr)^{\Omega(r)}$. %
\end{thm 3}

The lower bound for local Lipschitz filters in \cite{AJMR15} applies only to filters that are guaranteed to output a Lipschitz function (or a Lipschitz function with small error in values) for all random seeds. They can still have a constant error probability, but the only mode of failure allowed is returning a function that is far from $f$ when $f$ already satisfied the property. %
 Awasthi et al.~\cite{AJMR15} 
 suggest as a future research direction
 to consider local filters whose output does not satisfy the desired
property $\cP$ with small probability and mention that their techniques do not work for this case. We overcome this difficulty by providing different techniques that work for filters whose output fails to satisfy the Lipschitz property with constant error probability. In particular, \Cref{thm: lipschitz filter lower bound} applied with $r=\Theta(d)$ yields the lookup lower bound of $2^{\Omega(d)}$, as in \cite{AJMR15}, but without their restriction on the mode of failure, thus answering their open question. 

\subsubsection{Applications}\label{sec:intro-applications}
We showcase two applications of our local Lipchitz filters: to private data analysis and to tolerant testing. Both of them deal with real-valued functions with no a priori bound on the range.

\paragraph{Application to black-box privacy.} 
Our first application  is for providing differentially private mechanisms for releasing outputs of black-box functions, even in the presence of malicious clients.
 Differential privacy, introduced by~\cite{DworkMNS06}, is an accepted standard of privacy protection for releasing information about sensitive datasets. A sensitive dataset can be modeled as a point $x$ in $\{0,1\}^d$, representing whether each of the $d$ possible types of individuals is present in the data\footnote{A point $x\in\{0,1\}^d$ can also represent a dataset containing data of $d$ individuals, with one bit per individual: e.g., $x_i=1$ could indicate that the individual $i$ has a criminal record or some illness.}. More generally, it is modeled as a point $x$ in $[n]^d$ that represents a histogram counting the number of individuals of each type.

\begin{definition}[Differential Privacy]
\label{def: privacy} Two datasets $x,x'\in [n]^d$ are {\em neighbors}
if vertices $x,x'$ are neighbors in the hypergrid $\hypergrid.$  For privacy parameters $\eps>0$ and $\delta\in(0,1)$,
a randomized mechanism $\cM:[n]^d \to\R$ is {\em $(\epsilon,\delta)$-differentially private} if, for all neighboring $x,x'\in[n]^d$ and all measurable sets $Y\subset\R$, 
$$\Pr[\cM(x)\in Y]\leq e^\epsilon\Pr[\cM(x')\in Y]+\delta.$$
When $\delta=0$, 
then we call $\cM$
{\em purely differentially private}; otherwise, it is {\em approximately differentially private.} 
\end{definition}

A statistic (or any information about the sensitive dataset) is modeled as a function $f(x).$
One of the most commonly used building blocks in the design of differentially private algorithms is the Laplace mechanism\footnote{Gaussian mechanism is another popular differentially private algorithm that calibrates noise to the Lipschitz constant of $f$; for simplicity, we focus on the Laplace mechanism as a canonical example.}. The Laplace mechanism computing on a sensitive dataset $x$ can approximate the value $f(x)$ for a desired $c$-Lipschitz function $f$ by adding Laplace noise proportional to $c$ to the true value $f(x)$. The noisy value is safe to release while satisfying differential privacy.

 Multiple systems that allow analysts to make queries to a sensitive dataset while satisfying differential privacy have been implemented, including PINQ \cite{McSherry10}, Airavat~\cite{airavatNSDI10}, Fuzz~\cite{HaeberlenPN11}, and PSI~\cite{GaboardiHKNUV16}. They all allow releasing approximations to (some) real-valued functions of the dataset. In these implementations, the client sends a program to the server, requesting to evaluate it on the dataset, and receives the output of the program with noise added to it. 
 The program $f$ can be composed from a limited set of trusted built-in functions, such as sum and count. In addition, $f$ can use a limited set of (untrusted) data transformations, such as 
 combining several types of individuals into one type, 
 whose Lipschitzness can be enforced using programming languages tools.

 The limitation of the existing systems is that the functionality of the program is restricted by the set of trusted built-in functions available and the expressivity of the programming languages tools. Ideally, future systems would allow analysts to query arbitrary functions, specified as a black box. One reason for the black-box specification is to allow the clients to construct arbitrarily complicated programs. Another reason is to allow researchers analyzing sensitive datasets to obfuscate their programs in order to hide what analyses they are running from the data curator and their competitors. 
 
 The difficulty with allowing general queries is that when $f$ (supplied by a distrusted client) is given as a general purpose program, it is hard to compute its least Lipschitz constant, or even an upper bound on it. The data curator can ask the client to supply the Lipschitz constant for the query function $f$. However, as noted in \cite{JhaR13}, even deciding if $f$ has Lipschitz constant at most $c$ is NP-hard for functions over the finite domains we study (if $f$ is specified by a circuit). 
 Applying the Laplace mechanism with $c$ smaller than a Lipschitz constant (if the client supplied incorrect $c$) would result in a privacy breach, while applying it with a generic upper bound on the least Lipschitz constant of $f$ would result in overwhelming noise. One reason that a client might supply incorrect $c$ is simply because analyzing Lipschitz constants is difficult even for specialists (see~\cite{CasacubertaSVW22} for significant examples of underestimation of Lipschitz constants in the implementations of the simplest functions, such as sums, in differentially private libraries). Another reason is that client could lie in order to gain access to sensitive information\footnote{
 We give a specific example
 when $f$ has domain $\zo^d$ and a small range. Suppose $x_i=1$ if the individual $i$ has some illness that would disqualify them from getting a good rate on insurance and 0 otherwise. Say, a client would like to determine the secret bit of individual $i$ in the dataset. The client can submit a function $f$ such that $f(x)=0$ if $x_i=0$ and $f(x)=10/\epsilon$ if $x_i=1$, obfuscated or with really complicated code. The range of the function is $[0,10/\epsilon]$. If the data curator (incorrectly) believes that the function is 1-Lipschitz and uses the Laplace mechanism (formally specified in~\Cref{lem:laplace}) with noise parameter $1/\epsilon$ then the client will be able to figure out the bit $x_i$ with high probability, violating the privacy of $i$ in the strongest sense.}.

To the best of our knowledge, no existing $(\eps,\delta)$-DP mechanisms for the black-box privacy problem simultaneously achieve a runtime that is polynomial in $d$ and $\log\frac1\delta$ while providing an accuracy guarantee that is comparable to the Laplace mechanism. 
One solution to the black-box privacy problem with the untrusted client can be obtained using the propose-test-release method of Dwork and Lei \cite{DworkL09} and is described in \cite{BookAlgDP}. %
However, the runtime of this mechanism is $d^{O(\frac1\eps\log\frac1\delta)}$ which, when $\delta=\frac{1}{\poly d}$ is $d^{\Omega(\log d)}$, and it remains the same even when the input function has bounded range\footnote{Since the mechanism stated in \cite[Ch. 7.3, algorithm 13]{BookAlgDP} computes the distance to the nearest ``unstable" point, the runtime is actually $n^d$. However, the following small modification suffices to obtain the runtime of $d^{O(\frac1\eps\log\frac1\delta)}$. When releasing a noisy ``distance to the nearest unstable instance" $\hat d$, one can add noise from a Laplace distribution truncated to $\pm\frac1\eps\log\frac1\delta$ instead of a regular Laplace distribution. As a result, the mechanism need only consider points at distance at most $\frac2\eps\log\frac1\delta$.}. Note that the propose-test-release method does not yield a local Lipschitz filter). A second solution to the black-box privacy problem was recently proposed in \cite{KohliL23}. They present a mechanism called ``TAHOE" which runs in time $d^{O(\frac{1}{\eps}\log\frac1\delta)}$ and outputs an answer with significantly more noise than the Laplace mechanism for $1$-Lipschitz functions. The advantage of TAHOE is that the algorithm need only query the function on subsets of the input. Another solution to the black-box privacy problem was proposed by Jha and Raskhodnikova~\cite{JhaR13} who designed the following {\em filter mechanism}:
{\em A client who does not have direct access to $x$ can ask the data curator for information about the dataset by specifying a Lipschitz\footnote{If the function is $c$-Lipschitz, it can be rescaled by dividing by $c$.} function $f$. The data curator can run a local filter to obtain a value $g(x)$, where $g$ is guaranteed to be Lipschitz. Then the curator can use the Laplace mechanism and release the obtained noisy value.}
 If the client is truthful (i.e., the function is Lipschitz), then, assuming that the local filter 
 gives access to $g=f$ 
 in the case that $f$ is already Lipschitz, 
 the accuracy guarantee of the filter mechanism is inherited from the Laplace mechanism. However, if the client is lying 
 about $f$ being Lipschitz, the filter ensures that privacy is still preserved. Observe that the running time and accuracy of the filter mechanism directly depends on the running time and accuracy of the local Lipschitz filter. The lower bound by Awasthi et al.~\cite{AwasthiJMR16} on the complexity of local Lipschitz filters implies $2^{\Omega(d)}$ running time for filter mechanisms, even for releasing functions of the form $f:\{0,1\}^d\to\R.$

We provide a mechanism for privately releasing outputs of black-box functions, even in the presence of malicious clients, in time that is quasi-polynomial in the dimension $d$ while providing accuracy comparable to the Laplace mechanism. For bounded-range functions, the running time of our mechanism is polynomial in $d$ and $\log \frac 1 \delta$.
We bypass the lower bound in~\cite{AwasthiJMR16} by using the filter mechanism for bounded-range functions repeatedly to simulate a noisy binary search. Our mechanism needs only query access to the input function, that is, it can be specified as a black box (e.g., as a complicated or obfuscated program).

 \newtheorem*{thm 4}{Theorem \ref{thm:binary search mechanism} (Binary search filter mechanism, informal)}
\begin{thm 4}
		For all $\eps > 0$ and $\delta \in (0,1)$, there exists an $(\eps, \delta)$-differentially private mechanism $\cM$ that gets lookup access to a function $f: [n]^d \to [0,r]$ %
        and has the following properties. %
        Let $\kappa=\log\min(r,nd)$.
\begin{itemize}
\item \textbf{Efficiency:} The lookup and time complexity of $\cM$ are $d^
{O(\frac1\eps\kappa\log\kappa)}
\polylog \frac n\delta$.
\item \textbf{Accuracy:}  If $f$ is Lipschitz, then for all $x\in [n]^d$, we have $\cM(x) \sim f(x)+\text{Laplace}(\frac\kappa\eps
)$ %
with probability at least 0.99.
\end{itemize}
\end{thm 4}

In \Cref{sec:privacy-application}, we state and prove more detailed guarantees for black-box privacy mechanisms. In particular, our guarantees are stronger for the case when the client can provide an upper bound $r$ on the range diameter that is significantly smaller than $nd$.
There are many bounded-range functions that are hard to implement with a fixed set of trusted functions (and thus they are not implemented in current systems). One primitive often used in differentially private algorithms is determining whether the secret dataset $x$ is far from a specified set $S$ (where ``far’’ means that many records in $x$ would have to change in order to obtain $S$). The set $S$ could capture datasets with the desired property or satisfying a certain hypothesis. For instance, $S$ is the set of datasets with no outliers in Brown et al.~\cite{BrownGSUZ21}.  
To solve this, the client could submit a function $f:[n]^d \to [0,10/\epsilon]$ that outputs $\max(\text{distance}(x,S), 10/\epsilon)$, where $\epsilon$ is the privacy parameter. The range of $f$ is $[0,10/\eps]$. Here $n$ and $d$ can be huge, whereas $\eps$ could be, say $1/10$ (a typical value of $\eps$ used in industry today is even larger than that). Since $f$ is 1-Lipschitz, it can be released with sufficient accuracy to determine whether $x$ is far from $S$ with high probability.

 Both of our filters can be used in the filter mechanism, and we show the type of resulting guarantees for bounded-range functions in \Cref{sec:privacy-application}. 
 The work of \cite{JhaR13} provides a Lipschitz filter as well, intended to be used in the filter mechanism. 
 The filter mechanism instantiated with their filter satisfies the stronger guarantee of \emph{pure} differential privacy, while performing $\Theta((\log n+1)^d)$ lookups per query. %
 In contrast, the filter mechanism instantiated with either of our filters uses only $\poly(d)$ lookups per query with constant-range functions, satisfies \emph{approximate} differential privacy, and has a stronger accuracy guarantee because of the distance-respecting nature of the filters. 
Whereas the accuracy guarantee of \cite{JhaR13} only holds when the client is honest about the function $f$ being Lipschitz, our distance-respecting filters provide an additional accuracy guarantee for ``clumsy clients'' that submit a function that is close to Lipschitz---on average over possible datasets, the error of the mechanism is proportional to $f$'s distance to the class of Lipschitz functions.

 Finally, %
 we  use the filter mechanism for bounded-range functions to construct a mechanism for arbitrary-range functions and prove \Cref{thm:binary search mechanism}.  
Since every Lipschitz function with domain $[n]^d$ can have image diameter at most $nd$, we can require that the client translate the range of their function to the interval $[0,nd]$. %
Observe that the range restriction $f(x)\in[0,nd]$ can be easily enforced locally\footnote{
Other natural assumptions (which can be viewed as promises on the function that filter gets and which one can potentially try to use to circumvent strong lower bounds for local Lipschitz filters) cannot be enforced as easily. Some examples are monotonicity (which is not easy to enforce and has been studied in the context of local filters~\cite{SaksS10,bgjjrw2012,AJMR15}) and $C’$-Lipschitzness (i.e., assuming the function is $C’$-Lipschitz and trying to enforce that it is $c$-Lipschitz for $c<C’$).
}, i.e., without evaluating $f$ at points other than $x$.
In order to privately release $f(x)$ at some $x\in[n]^d$ in time $\exp(\polylog(nd))$, we simulate a noisy binary search for the value of $f(x)$. The simulation answers queries of the form ``Is $f(x) > v$?’’, by clipping the range of $f$ to the interval $[v-r,v+r]$, where $r=\Theta(\frac1\eps\log nd)$, and running an instance of the filter mechanism on the clipped function to obtain a noisy answer $a(x)$. The noisy answer $a(x)$ can be interpreted as $f(x)>v$ if $a(x)>v+\frac1\eps\log r$; $f(x)\approx v$ if $a(x)\in [v-\frac1\eps\log r,v+\frac1\eps\log r]$; and $f(x)< v$, otherwise. The resulting noisy implementation of the binary search provides accurate answers when $f$ is Lipschitz and results in a mechanism that is always differentially private, no matter how the client behaves.

\paragraph{Application to tolerant testing.}
The second application we present is to tolerant testing of the Lipschitz property of real-valued functions on the hypercube domains. Tolerant testing, introduced in \cite{ParnasRR06} with the goal of understanding the properties of noisy inputs, is one of the fundamental computational tasks studied in the area of sublinear algorithms. Tolerant testing has been investigated for various properties of functions, including monotonicity, being a junta, and unateness \cite{FischerF05,ACCL07,FR10,BlaisCELR16,LeviW19,CanonneGG0W19,PallavoorRW22,BlackKR23}.

In the standard property testing terminology, a {\em property} $\cP$ is a set of functions. Given a parameter $\eps\in(0,1)$, a function $f$ is $\eps$-far from $\cP$ if at least an $\eps$ fraction of function values have to change to make $f\in\cP$; otherwise, $f$ is $\eps$-close to $\cP$. Given parameters $\eps_0,\eps\in(0,1)$ with $\eps_0<\eps$ and query access to an input function $f$, an {\em $(\eps_0,\eps)$-tolerant tester} for $\cP$ accepts with probability at least 2/3 if $f$ is $\eps_0$-close to $\cP$ and rejects with probability at least 2/3 if $f$ is $\eps$-far from $\cP.$ For the special case when $\eps_0=0$, the corresponding computational task is referred to as (standard) {\em testing}.

Testing of the Lipschitz property was introduced in \cite{JhaR13} and subsequently studied in \cite{CS13, DixitJRT13, BermanRY14, AwasthiJMR16, ChakrabartyDJS17,DixitRTV18,KalemajRV23}. Lipschitz testing of functions $f:\{0,1\}^d\to\R$ can be performed with $O(\frac d\eps)$ queries~\cite{JhaR13,CS13}. In contrast, prior to our work, no nontrivial tolerant tester was known for this property. As shown in \cite{FischerF05}, tolerant testing can have drastically higher query complexity than standard testing: some properties have constant-query testers, but no sublinear-time tolerant testers.

 As an application of our local filters,%
 we construct the first nontrivial tolerant Lipschitz tester (see \Cref{thm: lipschitz tolerant tester}) for functions $f:\zo^d\to\R$. 
\ifnum\compact=1
Our tester is $(\eps,2.01\eps)$-tolerant and has query and time complexity $\frac{1}{\eps^2}d^{O(\sqrt{d\log(d/\eps)})}$. We stress that this tester handles
\else
\newtheorem*{thm 5}{Theorem \ref{thm: lipschitz tolerant tester} (Tolerant tester)}
\begin{thm 5}
For all $\eps\in(0,\frac13)$ and all sufficiently large $d\in\N$, there exists an $(\eps,2.01\eps)$-tolerant tester for the Lipschitz property of functions on the hypercube $\hypercube$. The tester has query and time complexity $\frac{1}{\eps^2}d^{O(\sqrt{d\log(d/\eps)})}$.
\end{thm 5}
We stress that our tolerant tester can handle
\fi
functions with any range.
Given that Lipschitz functions on $\zo^d$ can have range $[0,d]$, and our $\ell_0$-filter has  time complexity $d^{O(r)}$ for functions $f:\zo^d\to[0,r]$, one might expect our tester to run in time $\exp(d)$ when no a priori upper bound on $f$ is available. We leverage additional structural properties of Lipschitz functions to reduce this to $\exp(\sqrt d)$. 

\subsection{Our Techniques}\label{sec:techniques}
\paragraph{Algorithms.} 
Our $\ell_1$-respecting filter for functions with range $[0,r]$ is essentially a local simulation of a new distributed algorithm that iteratively ``transfers mass'' from large function values to small function values, with each round reducing a bound on the Lipschitz constant by a factor of 2/3. 
Previous work \cite{AwasthiJMR16} used the idea of transferring mass for integer-valued functions on the hypergrid and transferred one unit along a single dimension in each iteration.
We transfer mass along an arbitrary matching, with an amount proportional to a bound on the Lipschitz constant; transferring any constant amount of mass would yield an $r$-round algorithm.
Transferring mass equal to 2/3 of the bound on the Lipschitz constant in each round preserves the invariant that the Lipschitz constant is at most $r \cdot (2/3)^t$ in round $t$, thus giving a bound of $O(\log r)$ on the number of rounds.

Our local Lipschitz filters leverage %
powerful advances in local computation algorithms (LCAs). Both of them are built on an LCA for  
obtaining a maximal matching based on Ghaffari's LCA \cite{ghaffari2022local} for maximal independent set, and they rely on the locality of the independent set algorithm to give lookup-efficient access to the corrected values. 
We run the maximal matching LCA on
the {\em violation graph} of a function $f$ with each edge labeled by the {\em violation score} $|f(x)-f(y)|-\dist_G(x,y)$, as developed in property testing~\cite{DGLRRS99,FischerLNRRS02,JhaR13,AwasthiJMR16}. To get efficient local filters for functions with range $[0,r]$, we take advantage of the fact that, for such functions, the maximum degree of the violation graph is at most $degree(G)^r$ and the fact that the  matching LCA has lookup complexity that is polynomial in the degree.

Our $\ell_1$-respecting filter runs in multiple stages. 
In each stage, it calls the maximal matching LCA on the current violation graph, which %
captures pairs of points with relatively large violation score. For each matched pair, the filter decreases the larger value and increases the smaller value by an amount proportional to the current bound on the violation score. 
This shrinking operation reduces the Lipschitz constant by a multiplicative factor, and does not increase the $\ell_1$-distance to the class of Lipschitz functions. Our $\ell_0$-respecting filter uses a different approach that only requires one stage. It
relies on a well known technique for computing a Lipschitz extension of any real-valued function with a metric space domain. Leveraging an LCA for maximal matching allows us to simulate this extension procedure locally.

We remark that  Lange, Rubinfeld and Vasilyan \cite{LangeRV23} used an LCA for maximal matching to correct monotonicity of Boolean functions. Their corrector fixes violated pairs by swapping their labels; however, this technique fails to correct Lipschitzness. Additionally, unlike the corrector of \cite{LangeRV23}, which may change a monotone function on a constant fraction of the domain, our filters guarantee that Lipschitz functions are never modified.

\paragraph{Lower bounds.}

The first idea in the proof of our lower bound for local filters is to reduce from the problem of distribution-free property testing. Our hardness result for this problem uses novel ideas for proving lower bounds for adaptive algorithms, typically a challenging task, for which the community has developed relatively few techniques. Specifically, we show that our construction allows an adaptive algorithm to be simulated by a nonadaptive algorithm with extra information and the same query complexity.
One of our main technical contributions  is 
a query lower bound for distribution-free Lipschitz testing of functions $f:\zo^d\to[0,r]$ that is exponential in $r$ and $\log(d/r)$ for any even $r$ satisfying $4\leq r\leq 2^{-16}d$. The lower bound we achieve demonstrates that our filters have %
nearly optimal query complexity.

Distribution-free testing 
--- property testing with respect to an arbitrary distribution $D$ on the domain using both samples from $D$ and queries to the input --- 
was first considered in~\cite{HalevyK07}. 
Lipschitz testing has been investigated with respect to uniform distributions~\cite{CS13, DixitJRT13, BermanRY14, AwasthiJMR16, ChakrabartyDJS17,DixitRTV18,KalemajRV23} and product distributions~\cite{DixitJRT13,ChakrabartyDJS17}, but not with respect to arbitrary distributions. %
Our lower bound demonstrates a stark contrast in the difficulty of Lipschitz testing with respect to arbitrary distributions compared to product distributions. In particular, the Lipschitz tester of~\cite{ChakrabartyDJS17} for functions $f:\{0,1\}^d\to\R$ has query complexity linear in $d$ for all product distributions,
whereas our lower bound for distribution-free Lipschitz testing (formally stated in \Cref{thm: distribution-free lower bound}) implies that a query complexity of $2^{{\Omega}(d)}$ is unavoidable for arbitrary distributions.
To prove our lower bound for distribution-free testing, we start by constructing two distributions, on positive and negative instances of this problem, respectively. The instances consist of a pair $(f,U)$, where $f:\{0,1\}^d\to[0,r]$ is a function on the hypercube and $U$ is a uniform distribution over an exponentially large set of points called {\em anchor points}. The anchor points come in pairs $(x,y)$ such that $x$ and $y$ are at distance $r$ for the positive distribution and distance $r-1$ for the negative. The function values are set to $f(x)=0$ and $f(y)=r$. For the points not in the support of $U$, the values are chosen to ensure that the Lipschitz condition is locally satisfied around the anchor points, and then the remaining values are set to $r/2$. We note that %
\cite{HalevyK05} 
also uses a %
construction involving pairs of anchor points to prove query and sample complexity lower bounds for distribution-free monotonicity testing; however, our approach introduces a novel ``simulation" technique for proving lower bounds on the query complexity of adaptive algorithms.

The crux of the proof of  \Cref{thm: distribution-free lower bound} is demonstrating that every deterministic (potentially {\em adaptive}) tester $\cT$  with insufficient sample and query complexity distinguishes the two distributions only with small probability. (By the standard Yao's principle this is sufficient.)  An algorithm is called {\em nonadaptive} if it prepares all its queries before making them. A general (adaptive) algorithm, in contrast, can decide on queries based on answers to previous queries. One of the challenges in proving that the two distributions  are hard to distinguish for $\cT$ is dealing with adaptivity. We overcome this challenge by showing that $\cT$ can be simulated by a {\em nonadaptive algorithm $\cT_{na}$ that is provided with extra information}. Specifically, it gets one point from every pair of grouped anchor points.
One of the key ideas in the analysis is that our hard distributions, and the sampling done by the tester, can be simulated by first obtaining the information provided to $\cT_{na}$
using steps which are identical for the two hard distributions, and only then selecting the remaining anchor points to obtain the full description of the function $f$ and the distribution $U$. It allows us to show that, conditioned on avoiding a small probability bad event, $\cT$ cannot distinguish the distributions.

\paragraph{Applications.} Our main technical contribution to the two application areas we consider is realizing that they can benefit from local filters for bounded-range functions, even when the functions in the applications have unbounded range. For the privacy application, we obtain our differentially private mechanism for general real-valued functions provided by using our local filters to simulate a noisy binary search. For the application to tolerant testing, we use McDiarmid's inequality and the observation that our $\ell_0$-respecting  Lipschitz filter works even with partial functions.

\subsection{Preliminaries on Lipschitz Functions}
First, we define two important special families of graphs. We consider the hypercube $\hypercube$ with vertices $\zo^d$ and the hypergrid $\hypergrid$ with vertices $[n]^d$. For both of them, two vertices are adjacent if they differ by one in one coordinate and agree everywhere else.
Now, we give preliminaries on Lipschitz functions. When we discuss the range (or image) of functions, we often refer to its diameter. The \emph{diameter} of a closed and bounded $S\subset \R$ is $\max_{y\in S}(y)-\min_{y\in S}(y)$.  Let $G=(V,E)$ be an undirected graph and let $f:V\to\R$. 

\begin{definition}[$c$-Lipschitz functions]
Fix a constant $c>0$ and a graph $G=(V,E)$. A function $f:V\to\R$ is $c$-Lipschitz w.r.t.\ $G$ if  $|f(x)-f(y)|\leq c\cdot\dist_G(x,y)$ for all $x,y\in V$. A 1-Lipschitz function is simply referred to as Lipschitz.
    Let  $\Lip$ be the set of Lipschitz functions w.r.t.\ $G$.
\end{definition}

\begin{definition}[Distance to Lipschitzness]
    For all graphs $G=(V,E)$, functions $f:V\to\R$, distributions $D$ over $V$, and $b\in\{0,1\}$, define the $\ell_b$-distance to Lipschitz w.r.t.\ a distribution $D$ as $\ell_{b, D}(f,\Lip)=\min_{g\in \Lip}\|f-g\|_{b,D}$, 
    where 
 \ifnum\compact=1
\ $\|f- g\|_{0, D} = \Pr_{x \sim D}[f(x) \ne g(x)]; \ 
    \|f - g\|_{1, D} = \E_{x \sim D} [|f(x) - g(x)|].$
 \else
    \begin{align*}
\|f- g\|_{0, D} &= \Pr_{x \sim D}[f(x) \ne g(x)];\\
    \|f - g\|_{1, D} &= \E_{x \sim D} [|f(x) - g(x)|].          \end{align*}
    
\fi    
  When $D$ is the uniform distribution, we omit it from the notation. The definition of $\|f-g\|_0$ applies when $f$ and $g$ are partial functions.
\end{definition}

Next, we define the violation score of a pair of points, and the violation graph of a function.

\begin{definition}[Violated pair, violation score]\label{def: Violation Score}
For $x,y\in V$, let $\dist_G(x,y)$ denote the shortest path distance from $x$ to $y$ in $G$. A pair $(x,y)$ of vertices is {\em violated} with respect to $f$ if $|f(x)-f(y)|>\dist_G(x,y).$  The {\em violation score} of a pair $(x,y)$  with respect to $f$, denoted $VS_f(x,y)$, is
\ifnum\compact=1
$VS_f(x,y)=|f(x)-f(y)|-\dist_G(x,y)$
\else
	$$VS_f(x,y)=|f(x)-f(y)|-\dist_G(x,y)$$
 \fi
if $(x,y)$ is violated and 0 otherwise. We extend these definitions to partial functions $g:V\to \R\cup\{?\}$, where $?$ denotes an undefined value, by stipulating that if $x$ or $y$
is in $g^{-1}(?)$ then $(x,y)$ is not violated.

\end{definition}
\begin{definition}[Violation Graph]
\label{def: Violation Graph}
	The $\tau$-violation graph with respect to $f$ is a directed graph, denoted $B_{\tau,f}$, with vertex set $V$ and edge set $\{(x,y):\space VS_f(x,y)>\tau\text{ and } f(x)<f(y)\}$.
\end{definition}

\subsection{Preliminaries on Local Computation Algorithms}

First, we define a {\em local computation algorithm (LCA)} for a graph problem.
\begin{definition}[Graph LCA]
Fix $\delta\in(0,1)$. A graph LCA $\cA(x,\rho)$ is a randomized algorithm that gets adjacency list access\footnote{An adjacency list lookup takes a vertex $x$ and returns the set of vertices adjacent to $x$.}
to an input graph $G = (V,E)$, a query $x\in V$, and a random seed $\rho\in\zo^*$. 
For each $\rho$, the set of outputs $\{\cA(x, \rho) ~|~ x \in V\}$ is consistent with some object defined with respect to $G$, such as a maximal matching in $G$. The fraction of possible random strings for which $\cA$ fails (i.e., defines an object that does not satisfy the constraints) of the problem, is at most $\delta$.
\end{definition}

 We use an LCA for obtaining a maximal matching based on Ghaffari's LCA \cite{ghaffari2022local} for maximal independent set. The description of how to obtain an LCA for a maximal matching based on Ghaffari's result \cite{ghaffari2022local} is standard and appears, for example, in \cite{LangeRV23}.
\begin{theorem}
    [\cite{ghaffari2022local}]
\label{thm:matchings}
Fix $N, \degree_0\in\mathbb{N}$, and $\delta_0\in(0,1)$. There exists a graph LCA \textsc{GhaMatch} for the maximal matching problem for graphs with $N$ vertices and maximum degree $\degree_0$. Specifically, on input $x$, it outputs $y$ if $(x,y)$ or $(y,x)$ is in the matching,
and outputs $\perp$ if $x$ has no match.
\textsc{GhaMatch} uses a random seed of length $\poly(\degree_0 \cdot \log(N/\delta_0))$,
runs in time $
\poly(\degree_0 \cdot \log(N/\delta_0))$ per query, and has failure probability at most~$\delta_0$.
\end{theorem}

 We specify an LCA for accessing the violation graph.
 To simplify notation, we assume that any algorithm used as a subroutine gets access to the inputs of the algorithm which calls it; only the inputs that change in recursive calls are explicitly passed as parameters. 

\begin{algorithm}\label{alg:viol}
\textbf{Input}: Adjacency lists access to $G=(V,E)$, lookup access to $f:V \to [0,r]$, range diameter $r\in\R$, threshold $\tau\le r$, vertex $x \in V$\\
\textbf{Output}: Neighbor list of $x$ in $B_{\tau,f}$
\begin{algorithmic}[1]
	\caption{\label{alg: violations all-neighbors}LCA: \textsc{Viol}$(f(\cdot), \tau, x)$}
	\State \Return $\{y~|~ \dist_G(x,y) <  |f(x) - f(y)| - \tau\}$ \Comment{Compute by performing a BFS from $x$}
   
 \end{algorithmic}
\end{algorithm}

Local filters were introduced by Saks and Seshadhri~\cite{SaksS10} and first studied for Lipschitz functions by Jha and Raskhodnikova~\cite{JhaR13}. %
\begin{definition}[Local Lipschitz filter]
\label{def: local lipschitz filter}
     For all $c>0$ and $\delta\in(0,1),$  
     a {\em local $(c, \delta)$-Lipschitz filter\footnote{While the graph is hardcoded in this definition, our filters work when given adjacency list access to any graph.} over a graph $G=(V,E)$} is an algorithm $\mathcal{A}(x,\rho)$ that gets a query $x\in V$ and a random seed $\rho\in\{0,1\}^\star$, as well as lookup access to a function $f:V\to\R$ and adjacency lists access to $G$. With probability at least $1-\delta$ (over the random seed), the filter $\cA$ 
     provides query access to a $c$-Lipschitz function $g_\rho:V\to\R$ such that whenever $f$ is $c$-Lipschitz $g_\rho=f$. In addition, for all $\lambda>0$, the filter is $\ell_p$-respecting with blowup $\lambda$ if $\|f-g_\rho\|_p\leq\lambda\cdot\ell_p(f,\Lip)$ whenever $g_\rho$ is $c$-Lipschitz.
    \end{definition}

\section{$\ell_1$-respecting Local Lipschitz Filter}\label{sec:l1-filter}

The $d$-dimensional hypergrid of side length $n$ is the undirected graph, denoted $\hypergrid$, with the vertex set $[n]^d$ and the edge set $\{(x,y): |x-y|=1\}$.

\begin{theorem}
\label{thm: ell_1 LCA}
For all $\err>0$ and $\delta\in(0,1)$, there is an $\ell_1$-respecting local $(1 + \err,\delta)$-Lipschitz filter with blowup $2$ over the $d$-dimensional hypergrid $\hypergrid$. Given lookup access to a function $f:[n]^d \to [0,r]$, and a random seed $\rho$ of length $d^{O(r)}\cdot \polylog(n \log(r/\err)/\delta))$, the filter has lookup and time complexity $(d^r \cdot \polylog(n/\delta))^{O(\log (r/\err))}$ for each query $x\in[n]^d$.
If $f$ is Lipschitz, then the filter outputs $f(x)$ for all queries $x$ and random seeds. 
\end{theorem} 

We first give a global Lipschitz filter (\Cref{alg: ell_1 Filter}) and then show how to simulate it locally (in \Cref{alg: ell_1 Filter LCA}) by using the result of \cite{ghaffari2022local} stated in \Cref{thm:matchings}.

\subsection{Analysis of the Global Filter}

\begin{algorithm}
\textbf{Input}: Graph $G = (V, E)$, function $f:V \to [0,r]$, range diameter $r\in\R$, and 
\ifnum\compact=0
approximation 
\fi
parameter $\err>0$\\
\textbf{Output}: $(1 + \err)$-Lipschitz function $g: V \to [0,r]$
\begin{algorithmic}[1]
	\caption{\label{alg: ell_1 Filter} \textsc{GlobalFilter}$_1$}
    \State  Let $g_1\leftarrow f$
    \For{$t\gets 2$ to $\log_{3/2}(\frac{r}{\err})+1$}\Comment{Start at $t=2$ for $\textsc{GlobalFilter}_1$-$\textsc{LocalFilter}_1$ analogy.} 
        \State Set threshold $\tau \gets r \cdot (\frac{2}{3})^{t-1}$ and move-amount $\Delta \gets \frac{r}{3} \cdot (\frac{2}{3})^{t-2}$
	    \State Construct $B_{\tau,g_{t-1}}$ (\Cref{def: Violation Graph}) and compute a maximal matching $M_t$ of $B_{\tau,g_{t-1}}$ 
        \State Set $g_t \gets g_{t-1}$
        \For{$(x,y) \in M_t$}\Comment{Recall: $f(x)<f(y)$}
            \State Set $g_t(x) \gets g_t(x) + \Delta$
            \State Set $g_t(y) \gets g_t(y) - \Delta$
        \EndFor

        \EndFor
	\State \Return $g_{t}$, where $t=\log_{3/2}(\frac{r}{\err})+1$ 
 \end{algorithmic}
\end{algorithm}

The guarantees of $\textsc{GlobalFilter}$ (\Cref{alg: ell_1 Filter}) are summarized in the following lemma.
\begin{lemma}\label{lem: ell_1 Filter Correctness} 
	For all input graphs $G=(V,E)$, functions $f:V\to[0,r],$ and $\err>0$, if $g$ is the output of $\textsc{GlobalFilter}_1$, then $g$ is a $(1 + \err)$-Lipschitz function and $\|g-f\|_1\leq2\ell_1(f,\Lip)$.
\end{lemma}

We prove \Cref{lem: ell_1 Filter Correctness} via a sequence of claims.
\Cref{claim: reverse metric} makes an important observation about  the violation scores on adjacent edges in the violation graph. \Cref{claim: VS reduction} argues that the violation scores decrease after each iteration of the loop. \Cref{claim: lipschitz output} converts the guarantee for each iteration to the guarantee on the Lipschitz constant for the output function. Finally, \Cref{claim: ell_1 distance preserving} bounds the $\ell_1$-distance between the input and output functions.

\begin{claim}\label{claim: reverse metric} 

	If  $(x,y)$ and $(y,z)$ are edges in the violation graph $ B_{0,f}$ then $VS_f(x,z)\geq VS_f(x,y)+VS_f(y,z)$.
\end{claim}
\begin{proof}
Since $(x,y)$ and $(y,z)$ are edges in the violation graph $B_{0,f}$, then $f(x)<f(y)<f(z)$. Therefore,
	\begin{align*}
		VS_f(x,z)
		&=f(z)-f(x)-\dist_G(x,z)\\
		&\geq f(z)-f(y)+f(y)-f(x)-\dist_G(x,y)-\dist_G(y,z)\\
		&=VS_f(x,y)+VS_f(y,z),
	\end{align*}
by the definition of the violation score and the triangle inequality. 
\end{proof}

Next, we abstract out and analyze the change to the function values made in each iteration of the loop in \Cref{alg: ell_1 Filter}. 

\begin{claim}[Violation Score Reduction]
    \label{claim: VS reduction}
 Let $G=(V,E)$ be a graph and let $g:V\to [0,r]$ be a function such that $VS_g(x,y)\leq\vb$ for all $x,y\in V$. Let $M$ be a maximal matching in $B_{2\vb/3,g}$ such that $g(x)<g(y)$ for all $(x,y)\in M$. 
 Obtain $h$ as follows: set $h=g$ and then, for every edge $(x,y)\in M$, set $h(x)\gets h(x)+\vb/3$ and $h(y)\gets h(y)-\vb/3$.
 Then $VS_h(x,y)\leq 2\vb/3$ for all $x,y\in V.$
\end{claim}

\begin{proof}
Suppose $x,y\in V$, assume w.l.o.g.\ that $g(x)\leq g(y)$, and consider the following two cases. Recall that edges in the violation graph $(x,y)$ are directed from the smaller value $x$ to the larger value $y$.

{\bf Case 1:} $(x, y)$ is an edge in $B_{2\vb/3,g}$. By \Cref{claim: reverse metric}, we cannot have vertices $a,b,c$ such that $(a,b)$ and $(b,c)$ are both in $B_{2\vb/3,g}$, since otherwise $VS_g(a,c)$ would be at least $VS_g(a,b)+VS_g(b,c)> 2\vb/3+ 2\vb/3>\vb,$ contradicting the upper bound of $\vb$ on violation scores stated in \Cref{claim: VS reduction}. Thus, in $B_{2\vb/3,g}$, each edge incident on $x$ is outgoing and each edge incident on $y$ is incoming. Consequently, $h(x)\geq g(x)$ and $h(y)\leq g(y)$.

Moreover, since $(x,y)$ is in $B_{2\vb/3,g}$ and $M$ is a maximal matching, at least one of $x,y$ must participate in $M$. W.l.o.g.\ assume that $M$ contains an edge $(x,z)$. Then $h(x)=g(x)+\vb/3$. Since $h(y)\leq g(y)$, we have 
$VS_h(x,y)\leq VS_g(x,y)-\vb/3\leq 2\vb/3.$

{\bf Case 2:} $(x,y)$ is not an edge in $B_{2\vb/3,g}$. Then $VS_g(x,y)\leq 2v/3.$ Consider how the values of $x$ and $y$ change when we go from $g$ to $h$. Observe that $|h(x)-g(x)|$ is 0 or $\vb/3$ for all $x\in V$. If both values for $x$ and for $y$ stay the same, 
or move in the same direction (both increase or both decrease), then $VS_h(x,y)= VS_g(x,y)\leq 2\vb/3.$ {If they move towards each other, then $VS_h(x,y)\leq 2\vb/3$, whether $h(x)\leq h(y)$ or not.}

Now consider the case when they move {away from each other, that is, $h(x)\leq g(x)$ and $h(y)\geq g(y),$ and at least one of the inequalities is strict. First, suppose both inequalities are strict.}
Then there are vertices $z_x,z_y$ such that $(z_x,x),(y,z_y)\in M$. By \Cref{claim: reverse metric}, pair $(x,y)$ is not violated in $g$ (since otherwise  $VS_g(z_x,z_y)$ would be {at least}
$VS_g(z_x,x) +VS_g(x,y)+VS_g(y,z_y)>4\vb/3$, contradicting the 
assumption on violation scores in the claim). Since the values of the endpoints move by $\vb/3$ each, the new violation score $VS_h(x,y)\leq 2\vb/3.$

Finally, consider the case when only one of the {inequalities is strict.} W.l.o.g.\ suppose $h(y)>g(y).$ Then there is $z\in V$ such that $(y,z)\in M$. By \Cref{claim: reverse metric}, the violation score $VS_g(x,y)\leq \vb/3$, since otherwise $VS_g(x,z)$ would be at least $VS_g(x,y)+VS_g(y,z)>\vb/3+2\vb/3=\vb,$ contradicting the assumption on violation scores in the claim. Thus, $VS_h(x,y)\leq VS_g(x,y)+\vb/3\leq 2\vb/3,$ as required. 
\end{proof}

Next, we use \Cref{claim: VS reduction} to bound the Lipschitz constant of the function output by the global filter.

\begin{claim}\label{claim: lipschitz output}
    For all $t\geq 1$, the function $g_t$ computed in \Cref{alg: ell_1 Filter} is $(1+r(\frac{2}{3})^{t-1})$-Lipschitz. In particular, if $t^*=\log_{3/2}(r/\err)+1$ then $g_{t^*}$ is $(1 + \err)$-Lipschitz.
\end{claim}
\begin{proof}
    Fix a graph $G=(V,E)$ and a function $f:V\to [0,r]$. Then $VS_f(x,y)<r$ for all $x,y\in V$. 
    
    For all $t\geq 1$, let $\vb_{t}=r(\frac{2}{3})^{t-1}$. Notice that, in Line~\ref{line:assignment} of \Cref{alg: ell_1 Filter}, we set $\tau=\frac{2}{3}\vb_{t-1}$ and $\Delta=\frac{r}{3}(\frac{2}{3})^{t-2}=\frac{1}{3}\vb_{t-1}$. To prove the claim, it suffices to show that for all $t\geq 1$ and $x,y\in V,$ 
    \begin{align}\label{eq:induction}
VS_{g_t}(x,y)\leq \vb_t,
    \end{align}
     since then $|g_t(x)-g_t(y)|\leq\dist_G(x,y)(1+\vb_t)$. We prove \Cref{eq:induction} by induction on $t$.

    In the base case of $t=1$, we have $VS_{g_1}(x,y)<r=\vb_1$ for all $x,y\in V$. Assume \Cref{eq:induction} holds for some $t\geq 1$. Then, instantiating \Cref{claim: VS reduction} with $g=g_t, h=g_{t+1},$ and $\vb=\vb_t$ yields $VS_{g_{t+1}}(x,y)\leq2\vb_{t}/3=\vb_{t+1}$ for all $x,y\in V$.
\ifnum\compact=0

\fi
    In conclusion, since $\vb_{t^*}=\err$, the function $g_{t^*}$ is $(1 + \err)$-Lipschitz.
\end{proof}

Finally, we argue that the $\ell_1$-distance between the input and the output functions of \Cref{alg: ell_1 Filter} is small.

\begin{claim}
\label{claim: ell_1 distance preserving}
	Fix a graph $G=(V,E)$ and a function $f:V\to[0,r]$. Suppose $h$ is the closest (in $\ell_1$-distance) Lipschitz function to $f$. Then for all $t\geq 1$, the functions $g_t$ computed by \Cref{alg: ell_1 Filter} satisfy $\|g_{t+1}-h\|_1\leq\|g_{t}-h\|_1$ and $\|g_{t}-f\|_1\leq2\ell_1(f,\Lip)$.
\end{claim}
\begin{proof}
    Fix $t\geq 1$. Since $g_t$ and $g_{t+1}$ only differ on the endpoints of the edges in the matching $M_{t+1}$, we 
    restrict our attention to 
    those points. For each edge  $(x,y)\in M_{t+1}$, we will show
    \begin{align}\label{eq:edge-change}
        |g_{t+1}(x)-h(x)|+|g_{t+1}(y)-h(y)|\leq |g_t(x)-h(x)|+|g_t(y)-h(y)|.
    \end{align}
    
    Let $\tau=r(\frac{2}{3})^{t}$ and $\Delta=\frac{r}{3}(\frac{2}{3})^{t-1}=\frac{\tau}{2}$. Suppose $(x,y)\in M_{t+1}$. Recall that this implies that $VS_{g_t}(x,y)> \tau=2\Delta$ and $g_t(x)<g_t(y)$. By construction, $g_{t+1}(x)=g_t(x)+\Delta$ and $g_{t+1}(y)=g_t(y)-\Delta$. Thus, the violation score of $(x,y)$ decreased by $2\Delta$, so $(x,y)$ is still violated by $g_{t+1}$, i.e.,
    \begin{align}\label{eq:still-violated}
        g_{t+1}(x)+\dist_G(x,y)< g_{t+1}(y).
    \end{align}
    
    Define $\Phi(z)=|g_{t+1}(z)-h(z)|-|g_t(z)-h(z)|$ for all $z\in V$. (Intuitively, it captures how much further from $h(z)$ the value on $z$ moved when we changed $g_t$ to $g_{t+1}$.) Then \Cref{eq:edge-change} is equivalent to $\Phi(x)+\Phi(y)\leq 0$. If both $\Phi(x)\leq 0$ and $\Phi(y)\leq 0$, then \Cref{eq:edge-change} holds. Otherwise, $\Phi(x)>0$ or $\Phi(y)>0.$  Suppose w.l.o.g.\ $\Phi(x)>0.$ Since $g_{t+1}(x)=g_t(x)+\Delta,$ we know that $\Phi(x)\leq\Delta.$ To demonstrate that \Cref{eq:edge-change} holds, it remains to show that $\Phi(y)\leq -\Delta.$

    Since $\Phi(x)>0$, the value $h(x)$ is closer to $g_t(x)$ than to $g_{t+1}(x).$ Since $g_{t+1}(x) =g_t(x)+\Delta$, it implies that $h(x)$ must be below the midpoint between $g_t(x)$ and $g_{t+1}(x),$ which is $g_{t+1}(x)=\Delta/2$. That is,
    \begin{align}\label{eq:h-is-small}
        h(x)<g_{t+1}(x)-\Delta/2.
    \end{align}
We use that $h$ is Lipschitz, then apply Equations (\ref{eq:h-is-small}) and (\ref{eq:still-violated}) to obtain
\begin{align*}
    h(y)<h(x)+\dist_G(x,y)
    < g_{t+1}(x)-\Delta/2+\dist_G(x,y)
    <  g_{t+1}(y)-\Delta/2.
\end{align*}
Since $h(y)< g_{t+1}(y)-\Delta/2$ and $g_{t+1}(y)=g_t(y)-\Delta,$ we get that $g_t(y)$ and $g_{t+1}(y)$ are both greater than $h(y)$. Thus, $|g_{t+1}(y)-h(y)|=|g_t(y)-h(y)|-\Delta$ and hence, $\Phi(y)=-\Delta,$ so \Cref{eq:edge-change} holds.

We proved that \Cref{eq:edge-change} holds for 
every edge in $M_{t+1}$. Moreover, for all vertices $z$ outside of $M_{t+1}$, we have $g_{t+1}(z)=g_t(z)$ and, consequently, $\Phi(z)=0$. Summing over all vertices, we get that $\sum_{x\in V} \Phi(x)\leq 0.$ Thus, $\|g_{t+1}-h\|_1\leq\|g_t-h\|_1$. 
By the triangle inequality, $\|g_t-f\|_1\leq\|g_t-h\|_1+\|h-f\|_1\leq\|g_1-h\|_1+\|f-h\|_1=2\ell_1(f,\Lip)$.    
\end{proof}

\Cref{lem: ell_1 Filter Correctness} follows from Claims~\ref{claim: lipschitz output} and \ref{claim: ell_1 distance preserving}.

\subsection{Analysis of the Local Filter}\label{sec:ell_1 local filter}
In this section, we present a local implementation of \Cref{alg: ell_1 Filter} and complete the proof of \Cref{thm: ell_1 LCA}. We claim that for each $t\in[\log(r/\err)+1]$, \Cref{alg: ell_1 Filter LCA} simulates round $t$ of \Cref{alg: ell_1 Filter} and, for graphs on $N$ vertices with maximum degree $\degree$, has lookup complexity
$(\degree^r \cdot \polylog(N/\delta))^{O(\log (r/\err))}$.

\begin{algorithm}[H]
\textbf{Input}: Adjacency lists access to graph $G = (V,E)$, lookup access to $f: V\to [0,r]$, range diameter $r\in\R$, vertex $x \in V$, iteration number $t$, approximation parameter $\err>0$, and random seed $\boldsymbol{\rho} = \rho_1 \circ ... \circ \rho_{t}$\\
\textbf{Subroutines:} \textsc{GhaMatch} (see \Cref{thm:matchings}) and \textsc{Viol} (see \Cref{alg: violations all-neighbors})\\
\textbf{Output}: Query access to $(r \cdot (\frac{2}{3})^{t-1})$-Lipschitz function $g_{\boldsymbol{\rho}}: V \to [0,r]$
\begin{algorithmic}[1]
	\caption{\label{alg: ell_1 Filter LCA}LCA: \textsc{LocalFilter}$_1(x, t, \rho_1 \circ ... \circ \rho_{t})$}
	\If{$t=1$ or $r \cdot (\frac{2}{3})^{t-1} < \err$} 
       \State \Return $f(x)$
    \EndIf

    \State\label{line:assignment}Set threshold $\tau \gets r \cdot (\frac{2}{3})^{t-1}$ and move amount $\Delta \gets \frac{r}{3} \cdot (\frac{2}{3})^{t-2}$
    \State Set $f_{t}(x)\gets \textsc{LocalFilter}_1(x,t-1, \rho_1 \circ ... \circ \rho_{t-1})$
    \State\label{step:call-to-Viol} Set $y\gets \textsc{GhaMatch(Viol}(\textsc{LocalFilter}_1(\cdot,t-1, \rho_1 \circ ... \circ \rho_{t-1}),\tau,\cdot),x,\rho_t)$
    \If{$y \ne \bot$}
        \State $f_{t-1}(y)\gets \textsc{LocalFilter}_1(y,t-1, \rho_1 \circ ... \circ \rho_{t-1})$
		\State $f_t(x)\gets f_{t} (x) + \sign(f_{t-1}(y) - f_{t}(x)) \cdot \Delta$ 
    \EndIf
    \State\Return $f_t(x)$
 \end{algorithmic}
\end{algorithm}

\begin{definition}[Good seed]
\label{def: good matching random seed}
Let $G=(V,E)$ be a graph and fix $t\geq 1$. Consider a function $f:V\to [0,r]$. A string $\boldsymbol{\rho}=\rho_1 \circ ... \circ \rho_{t}$ is {\em a good seed for $G$ and $f$} if, for all $i\in[t]$, the matching computed by $\textsc{GhaMatch}$ in $\textsc{LocalFilter}_1(\cdot,i,\rho_1 \circ ... \circ \rho_i)$ is maximal. 
\end{definition}

\begin{claim}
\label{claim: ell_1 LCA correctness}
Fix a graph $G=(V,E)$, a function $f:V\to[0,r]$, and $\err>0$. Let $t^*=\log_{3/2}(r/\err)+1$ and fix a good seed $\boldsymbol{\rho}=\rho_1 \circ \dots \circ \rho_{t^*}$. Let $g(x)$ denote $\textsc{LocalFilter}_1(x,t^*, \boldsymbol{\rho})$ for all $x\in V$. Then $g$ is a $(1 + \err)$-Lipschitz function with range $[0,r]$ and $\|f-g\|_1\leq 2\ell_1(f,\Lip)$.

\end{claim}
\begin{proof}
    For all $t\in[t^*]$, let $g_t$ be the function computed by $\textsc{GlobalFilter}_1$ on $G,r,f,$ and $\err$ after iteration $t$ using the matching computed by the call to $\textsc{GhaMatch}$ in $\textsc{LocalFilter}_1(x,t,\rho_1 \circ ... \circ \rho_{t})$.
    Recall that the matching computed by each call to $\textsc{GhaMatch}$ in $\textsc{LocalFilter}_1(x,t,\rho_1 \circ ... \circ \rho_{t})$ is maximal and therefore can be used as the matching in the iteration $t$ of the loop in $\textsc{GlobalFilter}_1$.

    By an inductive argument, $\textsc{LocalFilter}_1(x,t,\rho_1 \circ \ldots\circ \rho_{t})=g_t(x)$ for all $x\in V$ and $t\in[t^*]$. The base case is $\textsc{LocalFilter}_1(x,1,\rho_1 )=f(x)=g_1$, and every subsequent $g_t$ computed by $\textsc{GlobalFilter}_1$ is the same as $\textsc{LocalFilter}_1(\sr{\cdot},t,\rho_1 \circ \ldots \circ \rho_{t})$.
    Hence, $\textsc{LocalFilter}_1(x,t^*,\boldsymbol{\rho})$ provides query access to $g_{t^*}$. By \Cref{lem: ell_1 Filter Correctness}, $g_{t^*}$ is $(1 + \err)$-Lipschitz and satisfies $\|g_{t^*}-f\|\leq2\ell_1(f,\Lip)$.
\end{proof}

\begin{lemma}
\label{lem: ell_1 LCA runtime}
Fix $\err>0$ and $\delta\in(0,1)$. Let $G=(V,E)$ be a graph with $|V|=N$ and  maximum degree $\degree$. Let $f:V\to[0,r]$ and $t^*=\log_{3/2}(r/\err)+1$. Then, for a random seed $\boldsymbol{\rho}=\rho_1 \circ ... \circ \rho_{t^*}$, which is a concatenation of $t^*$ strings of length $\degree^{O(r)}\polylog(Nt^*/\delta)$ each, the algorithm $\textsc{LocalFilter}_1(\cdot,t^*,\boldsymbol{\rho})$ is an $\ell_1$-respecting local $(1 + \err,\delta)$-Lipschitz filter with blowup $2$ and lookup and time complexity $(\degree^r \cdot \polylog(N/\delta))^{O(\log (r/\err))}$.
\end{lemma}

\begin{proof}
    Since the range of $f$ is at most $r$, two vertices in a violated pair can be at distance at most $r-1$. Hence, the maximum degree of the violation graph $B_{\tau,f}$ is at most $\degree^r$. By \Cref{thm:matchings} instantiated with $\degree_0=\degree^r$ and $\delta_0=\frac{t^*}{\delta}$, the failure probability of each call to \textsc{GhaMatch} is at most $\frac{\delta}{t^*}$. Since there are at most $t^*$ calls to \textsc{GhaMatch}, the probability that any call fails is at most $\delta$. It follows that a random string $\boldsymbol{\rho}$ of length specified in the lemma is a good seed (see \Cref{def: good matching random seed}) with probability at least $1-\delta$. 
    This allows us to apply \Cref{claim: ell_1 LCA correctness}, and conclude that $\textsc{LocalFilter}_1(x,t^*,\boldsymbol{\rho})$ provides query access to a $(1 + \err)$-Lipschitz function and fails with probability at most $\delta$ over the choice of $\boldsymbol{\rho}$.  
    
    Let $Q(t)$ be the lookup complexity of $\textsc{LocalFilter}_1(x,t,\rho_1 \circ ... \circ \rho_{t})$. Then $Q(1)=1$ and, since the max degree of $B_{\tau,f}$ is $\degree^r$, each lookup made by $\textsc{GhaMatch}$ to the violation graph oracle in the $(t-1)$-st iteration requires at most $\degree^rQ(t-1)$ lookups to compute. Since \textsc{GhaMatch} makes $\degree^{O(r)}\polylog(N/\delta)$ such lookups, $Q(t)\leq \degree^{O(r)}\polylog(n/\delta)Q(t-1)$. Thus, the final lookup complexity is $Q(t^*)\leq(\degree^r \cdot \polylog(N/\delta))^{O(\log (r/\err))}$. By inspection of the pseudocode, we see that the running time is polynomial in the number of lookups.
\end{proof}

\begin{proof}[Proof of \Cref{thm: ell_1 LCA}]
The theorem follows as a special case of \Cref{lem: ell_1 LCA runtime} with $G$ equal to the hypergrid $\hypergrid$. The hypergrid has $n^d$ vertices and maximum degree $2d$. This gives lookup and time complexity  $(d^r \cdot \polylog(n/\delta))^{O(\log (r/\err))}$.
If $f$ is Lipschitz, then all violation graphs are empty; therefore, any local matching algorithm returns an empty matching (or can otherwise be amended to do so by checking whether the returned edge is in the graph and returning $\perp$ if it is not). Thus, when $f$ is Lipschitz, the returned value is always $f(x)$.   
\end{proof}

\section{$\ell_0$-respecting Local Lipschitz Filter}\label{sec:ell0-filter}

In this section, we present a local Lipschitz filter that respects $\ell_0$-distance rather than $\ell_1$-distance. Unlike the $\ell_1$-respecting filter, the $\ell_0$-respecting filter outputs a function that is 1-Lipschitz.

\begin{restatable}{theorem}{mainthmtwo}
\label{thm: ell_0 LCA}
For all $\delta\in(0,1)$, there exists an $\ell_0$-respecting local $(1,\delta)$-Lipschitz filter with blowup $2$ over the $d$-dimensional hypergrid $\hypergrid$. Given lookup access to a function $f:[n]^d \to [0,r]$, and a random seed $\rho$ of length $d^{O(r)}\cdot \polylog(n/\delta))$, the filter has lookup and time complexity $d^{O(r)}\cdot \polylog(n/\delta)$ for each query $x\in[n]^d$. If $f$ is Lipschitz, then the filter outputs $f(x)$ for all queries $x$ and random seeds. If for all $y\in[n]^d$ we have $|f(x)-f(y)|\leq |x-y|$ then the filter outputs $f(x)$.
\end{restatable} 

We give a global view (\Cref{alg: ell_0 Filter}) and prove its correctness before presenting a local implementation (\Cref{alg: ell_0 Filter LCA}). We use the convention that $\max_{y\in S}(\cdot)$ is defined to be zero when $S$ is the empty set.

\subsection{Analysis of the Global Filter}

\begin{algorithm}[H]
\textbf{Input}: Graph $G = (V, E)$, function $f:V \to [0,r]$\\
\textbf{Output}: Lipschitz function $g: V \to [0,r]$
\begin{algorithmic}[1]
	\caption{\label{alg: ell_0 Filter} \textsc{GlobalFilter}$_0$}
	\State Construct $B_{0,f}$ (see \Cref{def: Violation Graph}) and compute a vertex cover $C$ of $B_{0,f}$
    \State Set $g_C\leftarrow f$
    \For{every vertex $u \in C$}
        \State Set $g_C(u)\leftarrow \max(0, \max_{v\in V\setminus C}(g_C(v)-\dist_G(u,v)))$
    \EndFor
	\State \Return $g_C$ 
 \end{algorithmic}
\end{algorithm}

\Cref{alg: ell_0 Filter} reassigns the labels on a vertex cover $C$ of the violation graph $B_{0,f}$. Observe that the partial function $f$ on the domain $V \setminus C$ is Lipschitz w.r.t.\ $G$, because its violation graph has no edges. We claim that this algorithm extends this partial function to a Lipschitz function defined on all of $G$. It is well known that for a function $f:X\to\R$ with a metric space domain,  
if $f$ is Lipschitz on some subset $Y\subset X$, then $f$ can be made Lipschitz while only modifying points in $X\setminus Y$. See, for example, \cite{JhaR13} and \cite{BookGeomNFA}. We include a proof for completeness.

\begin{claim}[Lipschitz extension]
    \label{claim:standard-extension}
    Let $G=(V,E)$ be a graph, and $f:V\to[0,r]$ a function. Then, for all vertex covers $C$ of $B_{0,f}$, the function $g_C$ returned by \Cref{alg: ell_0 Filter} is Lipschitz.
\end{claim}
\begin{proof}
Let $f:V\to[0,r]\cup \{?\}$ be a partial Lipschitz function and let $A_f$ be the set of points on which $f$ is defined. Fix a vertex $x\not\in A_f$ and obtain the function $g$ as follows: Set $g(y)=f(y)$ for all $y \in A_f$. Set $g(x) =\max_{v\in A_f}(f(v)-\dist_G(x,v)))$. Note that in the case where $A_f$ is empty, the function $f$ is nowhere defined, and hence setting $g(x)=0$ will always result in a Lipschitz function. Thus, assume w.l.o.g.\ that $A_f$ is not empty. 

    We will first argue that $g$ is Lipschitz. Let $v^*=\arg\max_{v\in A_f}(f(v)-\dist_G(x,v))$, i.e., a vertex such that $g(x) = f(v^*) - \dist_G(x, v^*)$.
	Then, for all $v\in A_f$,
    $$g(x)-g(v)=f(v^*)-\dist_G(x,v^*)-f(v)\leq \dist_G(v,v^*)-\dist_G(x,v^*)\leq \dist_G(x,v).$$
    Similarly, 
    $g(v)-g(x)\leq f(v)+\dist_G(x,v)-f(v)=\dist_G(x,v)$,
    so $g$ is Lipschitz. Notice that if $g$ is a Lipschitz function, then $\max(0,g)$ is also a Lipschitz function (truncating negative values can only decrease the distance between $g(x)$ and $g(y)$ for all pairs $x,y$ in the domain). Thus, setting $g(x)=\max(0,\max_{v\in A_f}(f(v)-\dist_G(x,v)))$ will also yield a Lipschitz function. 

    Next, we argue that the order of assignment does not affect the extension.
    Let $A_g$ be set of points on which $g$ is defined and note that $A_g=A_f\cup\{x\}$.
	We will show that for all $z\not\in A_g$ we have
    $$\max(0,\max_{v\in A_f}(f(v)-\dist_G(z,v)))=\max(0,\max_{v\in A_g}(g(v)-\dist_G(z,v))).$$
	Let $f(z)=\max(0,\max_{v\in A_f}(f(v)-\dist_G(z,v)))$ and $g(z)=\max(0,\max_{v\in A_g}(g(v)-\dist_G(z,v)))$. Then, since $A_g=A_f\cup\{x\}$, we obtain $g(z)=\max(f(z),g(x)-\dist_G(x,z))$. If $g(x)=0$ then $f(z)=g(z)$ since by definition $f(z)\geq 0$. On the other hand, if $g(x)>0$ then $g(x)=f(v^*)-\dist_G(x,v^*)$ and hence, 
    \begin{align*}
        f(z)&\geq f(v^*)-\dist_G(v^*,z)
        \geq(f(v^*)-\dist_G(x,v^*))-\dist_G(x,z)\\
        &=g(x)-\dist_G(x,z),
    \end{align*}
which implies $f(z)=g(z)$. To complete the proof of \Cref{claim:standard-extension}, Let $f:V\to[0,r]$ be a function with violation graph $B_{0,f}$. Notice that if $C$ is a vertex cover of $B_{0,f}$ then $f:V\setminus C\to[0,r]$ is Lipschitz, and thus, setting $f(x)=?$ for all $x\in C$ and applying the extension procedure inductively, we see that $g_C$ is a Lipschitz function.
\end{proof}

The following claim relates the distance to Lipschitzness to the vertex cover of the underlying graph. This relationship is standard for the Lipschitz and related properties, such as monotonicity over general partially ordered sets \cite[Corrolary 2]{FischerLNRRS02}. See \cite[Theorem 5]{CS13} for the statement and proof for the special case of hypergrid domains. The arguments in these papers extend immediately to the setting of general domains. 

\begin{claim}[Distance to Lipschitz]
\label{claim:l0 distance}
For all graphs $G = (V,E)$ on $n$ vertices and functions $f: V \to \R$, the size of the minimum vertex cover of the violation graph $B_{0,f}$ is exactly $n\cdot \ell_0(f, \Lip)$. 
\end{claim}
\begin{proof}[Proof of \Cref{claim:l0 distance}]
Let $C$ be any minimum vertex cover of $B_{0,f}$. We first argue that $n\cdot \ell_0(f, \Lip) \le |C|$. If $g$ is a partial function that is equal to $?$ on vertices in $C$ and equals $f$ elsewhere, then $g$ is Lipschitz outside of $C$. Then, by \Cref{claim:standard-extension} with $C$ as the cover in \Cref{alg: ell_0 Filter}, there exist values $y_1,\dots,y_{|C|}$ such that setting $g(x_i)=y_i$ for each $x_i\in C$ yields a Lipschitz function. It follows that $n\cdot \ell_0(f, \Lip) \le |C|$.

Next, using the function $g$ defined in the previous paragraph, suppose the set $P=\{x: g(x)\neq f(x)\}$ is not a vertex cover for $B_{0,f}$. Then, there exists some edge $(x,y)$ in $B_{0,f}$ such that $f(x)=g(x)$ and $f(y)=g(y)$ (i.e. $x,y\not\in P$). But by definition of $B_{0,f}$, this implies that $|g(x)-g(y)|>\dist_G(x,y)$ which contradicts the fact that $g$ is Lipschitz. It follows that $n\cdot \ell_0(f, \Lip)=|C|$.
\end{proof} 

\subsection{Analysis of the Local Filter}

Using \Cref{alg: violations all-neighbors} and \textsc{GhaMatch} from \Cref{thm:matchings}, we construct \Cref{alg: ell_0 Filter LCA}, an LCA which provides query access to a Lipschitz function close to the input function. It is analyzed in \Cref{lem:l0 LCA}.

\begin{algorithm}[H]
\textbf{Input}: Adjacency lists access to graph $G =(V,E)$, lookup access to $f: V\to [0,r]$, range diameter $r\in\R$,
vertex $x \in V$, random seed $\rho$\\
\textbf{Subroutines:} \textsc{GhaMatch} (see \Cref{thm:matchings}) and \textsc{Viol} (see \Cref{alg: violations all-neighbors}).\\
\textbf{Output}: Query access to Lipschitz function $g: V \to [0,r]$.
\begin{algorithmic}[1]
	\caption{\label{alg: ell_0 Filter LCA}LCA: \textsc{LocalFilter}$_0(x, \rho)$}
    \If{$\textsc{GhaMatch}(\textsc{Viol}(f, 0, \cdot), x, \rho) = \bot$}
        \State \Return $f(x)$
    \Else 
        \State $S \gets \{y ~|~ \dist_G(x,y) \le r \text{ and } \textsc{GhaMatch}(\textsc{Viol}(f, 0, \cdot), y, \rho) = \bot\}$

        \State \Return $\max(0,\max_{y\in S}(f(y) - \dist_G(x,y)))$
    \EndIf
 \end{algorithmic}
\end{algorithm}
\begin{lemma}[$\textsc{LocalFilter}_0$]
    \label{lem:l0 LCA}
    Fix $\delta\in(0,1)$. Let $G = (V,E)$ be a graph with $N$ vertices and maximum degree $\degree$. Then, for a random seed $\rho$ of length $\degree^{O(r)} \cdot \polylog(N/\delta)$, the algorithm $\textsc{LocalFilter}_0(x,\rho)$ (\Cref{alg: ell_0 Filter LCA}) is an $\ell_0$-respecting local $(1,\delta)$-Lipschitz filter with blowup $2$ and lookup and time complexity $\degree^{O(r)} \cdot \polylog(N/\delta)$.
\end{lemma}
\begin{proof}
    Since the range of $f$ is bounded by $r$, a pair of violated vertices $x,y$ must have $\dist_G(x,y)<r$, and thus, the maximum degree of $B_{0,f}$ is at most $\degree^r$. By \Cref{thm:matchings} instantiated with $\degree_0=\degree^r$ and $\delta_0=\delta$, the algorithm \textsc{GhaMatch} has lookup and time complexity $\degree^{O(r)} \cdot \polylog(N/\delta)$ per query, and fails to provide query access to a maximal matching with probability at most $\delta$ over the choice of $\rho$. Since $\textsc{LocalFilter}_0$ makes at most $\degree^r$ queries to \textsc{GhaMatch} and only fails when \textsc{GhaMatch} fails, $\textsc{LocalFilter}_0$ has lookup and time complexity $\degree^{O(r)} \cdot \polylog(N/\delta)$ and failure probability at most $\delta$.   
    Let $\rho$ be a seed for which \textsc{GhaMatch} does not fail, and let $C$ be the set of vertices that are matched by \textsc{GhaMatch} when given adjacency lists access to $B_{0,f}$. Since the matching is maximal, $C$ is a 2-approximate vertex cover of $B_{0,f}$. Hence, we can run $\textsc{GlobalFilter}_0$ (\Cref{alg: ell_0 Filter}) and use $C$ as the vertex cover. Since $\textsc{LocalFilter}_0$ and $\textsc{GlobalFilter}_0$ apply the same procedure to every vertex in $V$, and since  $C$ has at most twice as many vertices as a minimum vertex cover, Claims \ref{claim:standard-extension} and \ref{claim:l0 distance} imply that $\textsc{LocalFilter}_0$ provides query access to some Lipschitz function $g$ satisfying $\|f-g\|_0=|C|\leq2\ell_0(f,\Lip)$.
\end{proof}

\begin{proof}[Proof of \Cref{thm: ell_0 LCA}]
This is an application of \Cref{lem:l0 LCA} to the $d$-dimensional hypergrid $\hypergrid.$ The hypergrid has $n^d$ vertices and a maximum degree of $2d$, therefore, the lookup and time complexity are $d^{O(r)} \cdot \polylog(n^d/\delta) = d^{O(r)} \cdot \polylog(n/\delta)$. Similarly, the length of the random seed is also $d^{O(r)} \cdot \polylog(n/\delta)$.
If $f$ is Lipschitz, then all violation graphs are empty; therefore, any local matching algorithm returns an empty matching (or can otherwise be amended to do so by checking whether the returned edge is in the graph and returning $\perp$ if it is not). Thus, when $f$ is Lipschitz, the returned value is always $f(x)$.

If for all $y\in[n]^d$ we have $|f(x)-f(y)|\leq |x-y|$ then no edges in the violation graph are incident on $x$. Therefore, every local matching algorithm returns a matching that does not contain $x$, or can otherwise be amended to do so by checking whether the returned edge is in the graph and returning $\perp$ if it is not. Hence, the returned value is always $f(x)$.
\end{proof}

\section{Lower Bounds}\label{sec:lb}

In this section, we prove our lower bound for local filters, stated in \Cref{sec:intro}. We start with a more detailed statement of the lower bound. 

\begin{theorem}[Local Lispchitz filter lower bound]
\label{thm: lipschitz filter lower bound}
For all local $(1,\frac{1}{4})$-Lipschitz filters $\cA$ over the hypercube $\hypercube$, for all even $r\geq 4$ and integer $d \geq 2^{16}r$, there exists a function $f:\zo^d\to[0,r]$ for which the lookup complexity of $\cA$ is $(\frac dr)^{\Omega(r)}$. %
\end{theorem}

Note that the same bound on lookup complexity applies to local $(1+\frac{1}{2r},\frac{1}{4})$-Lipschitz filters over $\hypercube$. This can be seen by observing the Lipschitz constant in the hard distributions constructed in the proof of the theorem (in \Cref{def: hard distributions}). Our proof is via a reduction from distribution-free testing.
In \Cref{sec:filter-lb-derivation}, we state our lower bound on distribution-free testing of Lipschitz functions and use it to derive the lower bound on local Lipschitz filters stated in \Cref{thm: lipschitz filter lower bound}. In \Cref{sec:distribution-free-testing-lb}, we prove our lower bound for distribution-free testing.

\subsection{Testing Definitions and the Lower Bound for Local Lipschitz Filters}\label{sec:filter-lb-derivation}

We start by defining distribution-free testing of the Lipschitz property.

\begin{definition}[Distribution-free Lipschitz testing]
\label{def: D-free tester}
Fix $\eps\in(0,1/2]$ and $r\in\R$. A {\em distribution-free Lipschitz $\eps$-tester} $\cT$ is an algorithm that gets query access to the input function $f:\zo^d \to [0,r]$ and sample access to the input distribution $D$ over $\zo^d$. If $f$ is Lipschitz, then $\cT(f,D)$ accepts with probability at least $2/3$, and if $\ell_{0,D}(f, \mathcal{L}ip(\hypercube)) \ge \eps$, then it rejects with probability at least $2/3$.  
\end{definition}

We give a sample and query lower bound for this task.
\begin{theorem}[Distribution-free testing lower bound]
\label{thm: distribution-free lower bound}
Let $\cT$ be a distribution-free Lipschitz $\frac 12$-tester. Then, for all sufficiently large $d\in\mathbb{N}$ and even integers $4\leq r\leq 2^{-16}d$, there exists a function $f:\zo^d\to[0,r]$ and a distribution $D$, such that $\cT(f,D)$ either has sample complexity $2^{\Omega(d)}$, or query complexity $(\frac dr)^{\Omega(r)}$. %
\end{theorem}

Before proving \Cref{thm: distribution-free lower bound}, we use it to prove the lower bound on the lookup complexity of local Lipschitz filters, stated in \Cref{thm: lipschitz filter lower bound}. Recall that it says that every local $(1,\frac{1}{4})$-Lipschitz filter over the hypercube $\hypercube$ w.r.t.\ $\ell_0$-distance has worst-case lookup complexity $(\frac dr)^{\Omega(r)}$. %

\begin{proof}[Proof of \Cref{thm: lipschitz filter lower bound}]
Let $\cA$ be a local $(1,\frac{1}{4})$-Lipschitz filter for functions $f:\{0,1\}^d\to[0,r]$ over the hypercube $\hypercube$. Then, given an instance $(f,D)$ of the distribution-free Lipschitz testing problem with proximity parameter $\eps$, we can run the following algorithm, denoted $\cT(f,D)$: 
\begin{enumerate}
    \item Sample a set $S$ of $3/\eps$ points from $D$. 
    \item If $\cA(x,\rho)\neq f(x)$ for some $x\in S$ then {\bf reject}; otherwise, {\bf accept}.
\end{enumerate}

If $f$ is Lipschitz then, with probability at least $3/4>2/3$, we have $\cA(x,\rho)=f(x)$ for all $x \in S$ and, consequently, $\cT(f,D)$ accepts. Now suppose $f$ is $\eps$-far from Lipschitz with respect to $D$. Then $\cA$ fails with probability at most 1/4. With the remaining probability, it provides query access to some Lipschitz function $g_{\rho}$. This function disagrees with $f$ on a point sampled from $D$ with probability at least $\eps$. 
In this case, $\cT$ incorrectly accepts with probability at most $(1-\eps)^{3/\eps}\leq e^{-3}$. By a union bound, $\cT$ fails or accepts with probability at most $\frac 1 4+ e^{-3}\leq \frac 13.$ Therefore, $\cT$ satisfies \Cref{def: D-free tester}. By \Cref{thm: distribution-free lower bound}, 
$\cT$ needs at least $(\frac dr)^{\Omega(r)}$ %
queries, so $\cA$ must make at least $(\frac dr)^{\Omega(r)}$ lookups. %
\end{proof}

\subsection{Distribution-Free Testing Lower Bound}\label{sec:distribution-free-testing-lb}
We prove \Cref{thm: distribution-free lower bound} by constructing two distributions, $\pd$ and $\nd$, on pairs $(f,D)$ and then applying Yao's Minimax Principle \cite{Yao77}. We show that $\pd$ has most of its probability mass  on positive instances and $\nd$ has most of its mass on negative instances of distribution-free Lipschitz testing. The crux of the proof of  \Cref{thm: distribution-free lower bound} is demonstrating that every deterministic (potentially adaptive) tester with insufficient sample and query complexity distinguishes $\pd$ and $\nd$ only with small probability.

We start by defining our hard distributions. In both distributions, $D$ is uniform over a large set of points, called {\em anchor points}, partitioned into sets $A$ and $A'$, both of size $2^{d/64}.$ We treat $A$ (and $A'$), both as a set and as an ordered sequence indexed by $i\in[2^{d/64}]$. Points in $A$ and $A'$ with the same index are paired up; specifically, the pairs are $(A[i],A'[i])$ for all $i\in[2^{d/64}]$. For every point in $x\in \{0,1\}^d$ and radius $t>0$, let $\ball{x}{t}$ denote the \textbf{open} ball centered at $x$, that is, the set $\{y\in \{0,1\}^d: |x-y|< t\}$. For each point $x\in A$, the function value of every point $y\in \smallball{x}$ is equal to the distance from $x$ to $y.$  For each point $x\in A'$, the function value of every point $y\in \smallball{x}$
is equal to $r$ minus the distance from $x$ to $y$ where $r$ is the desired image diameter of the functions. The points in $A'$ are chosen so that every pair $(A[i],A'[i])$ satisfies the Lipschitz condition with equality in $\pd$ and violates the Lipschitz condition in $\nd$. 

\begin{definition}[Hard Distributions]
    \label{def: hard distributions}
    Fix sufficiently large $d\in\mathbb{N}$ and $4\leq r\leq 2^{-16}d$.
    For all $b\in\zo$, let $\bd$ be the distribution given by the following sampling procedure:
    \begin{enumerate}
        \item Sample a list $A$ of $2^{d/64}$ elements in $\zo^d$ independently and uniformly at random. 
        \item Sample a list $A'$ of the same length as $A$ as follows. For each $i\in[2^{d/64}]$, pick the element $A'[i]$ uniformly and independently from
        $\{y \in \zo^d: |A[i]-y| = r-b\}$. 
        The elements of $A\cup A'$ are called {\em anchor points}. Additionally, for each $i\in[2^{d/64}]$, we call $A[i]$ and $A'[i]$ \emph{corresponding anchor points}.
        \item\label{step:define-f} Define $f:\zo^d\to[0,r]$ by 
        \[f(x) = 
        \begin{cases} 
        |x - A[i]| &  \text{if } x\in\smallball{A[i]}
        \text{ for some } i\in[2^{d/64}];\\  
        r - |x - A'[i]| & \text{else if } x\in\smallball{A'[i]} 
        \text{ for some } i\in[2^{d/64}];\\  
        r/2 & \text{otherwise}.
        \end{cases}
        \]
        
        \item\label{step:return-f-U} Output $(f, U)$, where $U$ is the uniform distribution over $A\cup A'$.
    \end{enumerate} 
\end{definition}

Next, we define a bad event $B$ that occurs with small probability and analyze the distance to Lipschitzness of functions arising in the support of distributions $\bd$, conditioned on $\overline{B}$. For a distribution $D$ and an event $E$, let  $D|_E$ denote the conditional distribution of a sample from $D$ given $E$.

\begin{lemma}[Distance to Lipschitzness]\label{lem:dist-to-Lip}
Let $B_0$ be the event that $|A[i]-A[j]|\leq d/4$ for some distinct $i,j\in[2^{d/64}]$. Then
\begin{enumerate}
\item $\Pr_{\bd}[B_0]\leq 2^{-d/32}$ for all $b\in\{0,1\}.$
\item  If $(f,U)\sim\pd|_{\overline{B_0}}$ then $f$ is Lipschitz, and if $(f,U)\sim\nd|_{\overline{B_0}}$ then 
$\ell_{0,U}(f, \mathcal{L}ip(\hypercube)) \ge \frac 12.$
\end{enumerate}
\end{lemma}
\begin{proof}
To prove Item 1, choose $x,y\in\zo^d$ by setting each coordinate to one independently with probability $p=\frac{1}{2}$. Let $\mu=\Exp_{x,y}[d(x,y)]=2dp(1-p)=\frac{d}{2}$. By Chernoff bound, $\Pr_{x,y}[d(x,y)\leq\frac{\mu}{2}]\leq e^{-\frac{\mu}{8}}\leq 2^{-d/16}$. There are at most $2^{{d}/{32}}$ pairs of points in $A$. By a union bound over all such pairs,  $\Pr[B_0]\leq 2^{-d/16}\cdot 2^{{d}/{32}}= 2^{-d/32}.$ To prove Item 2, recall that $r\leq 2^{-16}d$. Suppose that $B_0$ did not occur. If $b=0$ then $f$ is Lipschitz because balls
$\smallball{x}$ are disjoint for all anchor points $x$.
When $b=1$, every pair $(A[i],A'[i])$ violates the Lipschitz condition, so $f$ is $1/2$-far from Lipschitz w.r.t. $U$.
\end{proof}

\subsubsection{Indistinguishability of the Hard Distributions by a Deterministic Algorithm}\label{sec:indistinguishability}

Fix a deterministic distribution-free Lipschitz $\frac 12$-tester $\cT$ that gets access to input $(f,U)$, takes $s=\frac{1}{8}\cdot2^{d/128}$ samples from $U$ and makes $q$ queries to $f$. Since the samples from $U$ are independent (and, in particular, do not depend on query answers), we assume w.l.o.g.\ that $\cT$ receives all samples from $U$ prior to making its queries. One of the challenges in proving that the distributions $\pd$ and $\nd$ are hard to distinguish for $\cT$ is dealing with adaptivity. We overcome this challenge by showing that $\cT$ can be simulated by a {\em nonadaptive algorithm $\cT_{na}$ that is provided with extra information}. In addition to its samples: $\cT_{na}$ gets at least one point from each pair $(A[i],A'[i]),$ as well as function values on these points. Next, we define the extended sample given to $\cT_{na}$ and the associated event $B_1$ that indicates that the sample is bad.
We analyze the probability of $B_1$ immediately after the definition.

\begin{definition}[Sample set, extended sample, bad sample event $B_1$]
\label{def:sample-set-and-B1}
Fix $b\in\zo$ and sample $(f,U)\sim\bd$. Let $S$ denote the sample set of $\frac{1}{8}\cdot2^{d/128}$ points obtained i.i.d.\
from $U$ by the tester $\cT$. The {\em extended sample} $S^+$ is the set $S\cup\{A[i]: i\in[2^{d/64}]\wedge A[i]\not\in S\wedge A'[i]\not\in S\}$.
Let $S^-$ denote the set $(A\cup A')\setminus S^+$.
A set $S^+$ is {\em good} if all distinct $x,y\in S^+$ satisfy $|x-y|>d/5$ and {\em bad} otherwise. Define $B_1$ as the event that $S^+$ is bad.
\end{definition}

\begin{lemma}[$B_1$ bound]
    \label{lem: B_1 bound}
    Fix $b \in\{0,1\}$. Then $\Pr_{\bd,S}[B_1]\leq \frac{1}{30}$. 
\end{lemma}

\begin{proof} 
Recall the bad event $B_0$ from \Cref{lem:dist-to-Lip}. By the law of total probability, 
\begin{align}\label{eq:B0-bound}
\Pr_{\bd,S}[B_1]= \Pr_{\bd,S}[B_1 | B_0]\cdot \Pr_{\bd}[B_0]+\Pr_{\bd,S}[B_1 | \overline{B_0]}\cdot \Pr_{\bd}[\overline{B_0}]
\leq \Pr_{\bd}[B_0]+ \Pr_{\bd,S}[B_1 | \overline{B_0}].    
\end{align}
To bound $\Pr_{\bd,S}[B_1 |\overline{B_0}]$, observe that if $B_0$ did not occur, then all pairs $(x,y)$ of anchor points, except for the corresponding pairs, satisfy $|x-y|>\frac d 4-2r>\frac d 5$ because $r<d\cdot 2^{-16}$. In particular, it means that all anchor points are distinct. Now, condition on $\overline{B_0}$. Then event $B_1$ can occur only if both $A[i]$ and $A'[i]$ for some $i\in[2^{d/64}]$ appear in the extended sample $S^+$. By \Cref{def:sample-set-and-B1}, this is equivalent to the event that both $A[i]$ and $A'[i]$ for some $i\in[2^{d/64}]$ appear in $S$. Each pair of samples in $S$ is a pair of corresponding anchor points with probability at most $2^{-d/64}$. By a union bound over the at most  
$\frac{1}{64}\cdot 2^{d/64}$ pairs of samples taken for $S$, the probability that $A[i],A'[i]\in S$ for some $i$ is at most 
$\frac{1}{64} \cdot 2^{d/64} \cdot 2^{-d/64}=\frac{1}{64}$.
Hence, $\Pr_{\bd,S}[B_1 |\overline{B_0}]\leq\frac{1}{64}$.
The lemma follows from \Cref{eq:B0-bound} and \Cref{lem:dist-to-Lip}.
\end{proof}

\subsubsection{The Simulator}

One of the key ideas in the analysis is that our hard distributions, and the sampling done by the tester, can be simulated by first obtaining the set $S^+$ using steps which are identical for $b=0$ and $b=1$, and only then selecting points in $S^-$ to obtain the full description of the function $f$ and the distribution $U$. Next, we state the simulation procedure. Note that the first \ref{step:last-common-step} steps of the procedure do not use bit $b$, that is, are the same for simulating $\pd$ and $\nd.$

\begin{definition}[Simulator]
    \label{def:simulator}
    Fix $b\in\{0,1\}$. Let $\bdhat$ be the distribution given by the following procedure:
    \begin{enumerate}
    \item Sample a list $S^+$ of $2^{d/64}$ elements in $\zo^d$ independently and uniformly at random.    
    \item For each $i\in[2^{d/64}]$, do the following: if $i\leq\frac{1}{8}\cdot2^{d/128}$,
    then assign  $S^+[i]$ to either $A[i]$ or $A'[i]$ uniformly and independently at random; if $i>\frac{1}{8}\cdot2^{d/128}$ 
    then assign $S^+[i]$ to $A[i]$.

    \item\label{step:last-common-step} Proceed as in Step~\ref{step:define-f} of the procedure in \Cref{def: hard distributions} to set $f(x)$ for all $x\in S^+.$

    \item For each $i\in[2^{d/64}]$, pick the element $S^-[i]$ uniformly and independently from
        $\{y \in \zo^d: |S^+[i]-y| = r-b\}$. 
    Assign it to $A'[i]$ if $S^+[i]$ was assigned to $A[i]$ and vice versa.

    \item Proceed as in Step~\ref{step:define-f} of the procedure in \Cref{def: hard distributions} to set $f(x)$ for all $x\notin S^+$ and output $(f, U,S^+)$, where $U$ is the uniform distribution over $A\cup A'$.
    \end{enumerate} 
\end{definition}

\Cref{obs:simulator} states that, conditioned on $\overline{B_1}$, the simulator produces identical distributions on the extended sample $S^+$ and function values $f(x)$ on points $x\in S^+$, regardless of whether it is run with $b=0$ or $b=1.$ Moreover, conditioned on $\overline{B_1}$, it faithfully simulates sampling $(f,U)$ from $\bd$ and $S$ from $U$, and then extending $S$ to $S^+$ according to the procedure described in \Cref{def:sample-set-and-B1}.
\begin{observation}[Simulator facts]\label{obs:simulator}

Let $(f_b,U_b,S^+_b)\sim \bdhat|_{\overline{B_1}}$ for each $b\in\{0,1\}.$ Let $f(S^+)$ denote function $f$ restricted to the set $S^+$. Then the distribution of $(S^+_0,f(S^+_0))$ is identical to the distribution of $(S^+_1,f(S^+_1))$.

Now fix $b\in\{0,1\}.$ Sample $(f,U)\sim \bd|_{\overline{B_1}}$ and $S\sim U$. Then the distribution of $(f,U,S^+)$ is identical to the distribution of $(f_b,U_b,S^+_b)$.   
\end{observation}
\begin{proof}
    Since the first \ref{step:last-common-step} steps in \Cref{def:simulator} do not depend on $b$, the distribution of $(S^+_0,f(S^+_0))$ is the same as the distribution of $(S^+_1,f(S^+_1))$. Now, fix $b\in\zo$. Notice that in the procedure for sampling from $\bd$ (\Cref{def: hard distributions}), for all $i\in[2^{d/64}]$ the anchor point $A[i]$ is a uniformly random point in $\zo^d$, and the anchor point $A'[i]$ is sampled uniformly from $\{y: |y-A[i]|=r-b\}$. By the symmetry of the hypercube, the marginal distribution of $A'[i]$ is uniform over $\zo^d$. Hence, sampling $A'[i]$ uniformly at random from $\zo^d$ and then $A[i]$ uniformly at random from $\{y:|y-A'[i]|=r-b\}$ yields the same distribution. Thus, the distribution of anchor points is the same under $\bdhat$ as under $\bd$ for each $b\in\zo$. Now, conditioned on $\overline{B_1}$, the set $S^+$ contains exactly one anchor point from each corresponding pair. By the preceding remarks, we can assume the anchor point in $S^+$ was sampled uniformly at random from $\zo^d$, and the corresponding anchor point was sampled uniformly from the set of points at distance $r-b$. Thus, conditioned on $\overline{B_0}$, the distribution over $S^+$ and $S^-$ is the same as the distribution over $S^+_b$ and $S^-_b$. Since $(U,f)$ and $(U_b,f_b)$ are (respectively), uniquely determined by $(S^+,S^-)$ and $(S^+_b,S^-_b)$, the distribution of $(U,f,S^+)$ is the same as the distribution of $(U_b,f_b,S^+_b)$.
\end{proof}

\subsubsection{The Bad Query Event and the Proof that Adaptivity Doesn't Help}

\Cref{obs:simulator} assures us that distributions of $\{(x,f(x)): x\in S^+\}$ are the same (conditioned on $\overline{B_1}$) for both hard distributions. The function values $f(y)$ for $y\in \bigcup_{x\in S^+}\smallball{x}$ are determined by  $\{(x,f(x)): x\in S^+\}$. Moreover, the function values $f(y)$ are set to $r/2$ for all $y\notin\bigcup_{x\in(S^+\cup S^-)}\smallball{x}.$
So, intuitively, the tester can 
distinguish the distributions only if it queries a point in $\smallball{x}$ for some $x\in S^-.$ Our last bad event, introduced next, captures this possibility.

\begin{definition}[Revealing point, bad query event $B_\cT$]
\label{def:B_T}
Fix $b\in\zo$ and sample $(f,U)\sim\bd$. A point $x$ is {\em revealing} if $x\in\smallball{y}$ for some $y\in S^-$. Let $B_\cT$ be the event that $\cT$ queries a revealing point.
\end{definition}

To bound the probability of $B_\cT$, we first introduce the nonadaptive tester $\cT_{na}$ that simulates $\cT$ to decide on all of its queries. Tester $\cT_{na}$ gets query access to a function $f$ sampled from $\bd$, a sample $S\sim U$ and $\{(x,f(x)):x\in S^+\}$. Subsequently, in \Cref{claim: adaptivity}, we argue that if $\cT$ queries a revealing point then $\cT_{na}$ queries such a point as well. This implies that $\Pr_{\bd,S}[B_\cT]\leq \Pr_{\bd,S}[B_{\cT_{na}}]$, where $B_{\cT_{na}}$ is defined analogously to $B_\cT$, but for the tester $\cT_{na}$. Finally, in \Cref{lem: BcT bound}, we upper bound the probability of $B_{\cT_{na}}$ by first arguing that, conditioned on $\overline{B_1}$, the probability that $\cT_{na}$ queries a revealing point is small. Combining this fact with the bound on $\overline{B_1}$ yields an upper bound on $B_{\cT_{na}}$ and, consequently, $B_{\cT}$.

\begin{definition}[$\cT_{na}$]
    \label{def: T_na}
    Let $\cT_{na}$ be a nonadaptive  deterministic algorithm that gets query access to $f$ sampled from $\bd$, sets $S\sim U$ and $\{(x,f(x)):x\in S^+\}$, and selects its queries by simulating $\cT$ as follows:
    \begin{enumerate}
        \item Provide $S$ as the sample and answer each query $x\in\{0,1\}^d$ with $g(x)$ defined by
$$g(x) = 
        \begin{cases}
        |x - y| & \text{ if }  x\in\smallball{y} \text{ for some } y \in S^+ \text{ satisfying } f(y)=0;\\
        r - |x - y| & \text{ else if } x\in\smallball{y} \text{ for some }y \in S^+ \text{ satisfying } f(y)=r; \\
        r/2 & \text{ otherwise}.
        \end{cases}
        $$
        \item Let $x_1,\dots,x_q$ be the queries made by $\cT$ in the simulation. Query $f$ on $x_1,\dots,x_q$.
    \end{enumerate}
\end{definition}

\begin{claim}[Adaptivity does not help]
    \label{claim: adaptivity} 
    If $\cT$ queries a revealing point then $\cT_{na}$ queries a revealing point.
\end{claim}
\begin{proof}
    Let $x_1,\dots,x_q$ be the queries made by $\cT$ and $y_1,\dots,y_q$ be the queries made by $\cT_{na}$. Since $\cT$ is deterministic and the sample set $S$ is the same in both $\cT_{na}$ and in $\cT$, we have $x_1=y_1$. Assume $\cT$ queries a revealing point. Let $m\in[q]$ be the smallest index such that $x_m$ is revealing. By definition of $g$ and revealing point, $g(x_i)=f(x_i)$ for all $i\in[m-1]$. Consequently, $y_m=x_m$ and $\cT_{na}$ queries a revealing point.
    \end{proof}

\begin{lemma}[$B_\cT$ bound]
    \label{lem: BcT bound} Fix $b\in\zo$. There exists a constant $\alpha>0$ such that, for all sufficiently large $d$, if $\cT$ makes $q=2^{\alpha r\log(d/r)}$ queries then 
    $\Pr_{D_b,S}[B_\cT]<\frac{2}{30}.$
\end{lemma}

\begin{proof} 
By \Cref{claim: adaptivity}, $\Pr_{D_b,S}[B_\cT]\leq \Pr_{D_b,S}[B_{\cT_{na}}]$. Applying the law of total probability we obtain the inequality $\Pr_{D_b,S}[B_{\cT_{na}}]\leq \Pr_{D_b,S}[B_{\cT_{na}}|\overline{B_1}]+\Pr_{D_b,S}[B_1]$. Next, we compute $\Pr_{\bd,S}[B_{\cT_{na}}|\overline{B_1}],$ which is equal to $\Pr_{\bdhat}[B_{\cT_{na}}|\overline{B_1}],$ since by \Cref{obs:simulator}, the simulator faithfully simulates sampling $(f,U)\sim \bd$ and then obtaining $S\sim U$, conditioned on $\overline{B_1}$. By the principle of deferred decisions, we can stop the simulator after Step~\ref{step:last-common-step}, then consider queries from $\cT_{na}$, and only then run the rest of the simulator. Since $\cT_{na}$ is a nonadaptive $q$-query algorithm, it is determined by the collection $(x_1,\dots,x_q)$ of query points that it chooses as a function of its input (sets $S$ and $\{(x,f(x)):x\in S^+\}$). We will argue that  the probability (over the randomness of the simulator) that the set $\{x_1,\dots,x_q\}$ contains a revealing point (i.e., a point on which $f$ and $g$ disagree) is small. Consider some query $x$ made by $\cT_{na}$. By \Cref{def:B_T}, a point $x$ is  revealing if $x\in\smallball{S^-[i]}$ for some $i\in[2^{d/64}]$. Recall that each $S^-[i]$ satisfies
$|S^+[i]-S^-[i]| = r-b$, and thus each revealing point is in $\ball{S^+[i]}{3r/2}$ for some $i\in[2^{d/64}]$. 
All such balls around anchor points in $S^+$ are disjoint, because $r\leq 2^{-16} \cdot d$ and we are conditioning on $\overline{B_1}$ (the event that all pairs of points in $S^+$ are at distance greater than $d/5$). 

Suppose $x\in\ball{S^+[i]}{3r/2}$ for some $i\in[2^{d/64}]$. (If not, $x$ cannot be a revealing point.) For $x$ to be revealing, it must be in $\smallball{S^-[i]}$ or equivalently, $S^-[i]$ must be in $\smallball{x}$. The simulator chooses $S^-[i]$ uniformly and independently from $\{y \in \zo^d: |S^+[i]-y| = r-b\}$.
The number of points at distance $r-b\geq r-1$ from $S^+[i]$ is at least 
$\binom{d}{r-b}>(\frac{d}{r})^{r-1}\geq (\frac{d}{r})^{3r/4}$, where the last inequality holds because $r\geq 4$.
Out of these choices, only those that are in $\smallball{x}$ will make $x$ a revealing point. The number of points in $\smallball{x}$ is $\sum_{i=0}^{r/2}\binom{d}{i}\leq r\binom{d}{r/2}\leq r(\frac{2de}{r})^{r/2}$. Then
\begin{align*}
    \Pr_{\bdhat}[x\text{ is revealing}~|~\overline{B_1}]
    &\leq  r\Big(\frac{2de}{r}\Big)^{r/2}\Big(\frac{r}{d}\Big)^{3r/4}
     \leq r(2e)^{r/2}\Big(\frac{r}{d}\Big)^{r/4}
     =2^{\log(r)+\frac{r}{2}\log(2e)-\frac{r}{4}\log(d/r)}\\
     &\leq 2^{2r-\frac{r}{4}\log(d/r)}
     \leq 2^{-\frac{r}{8}\log(d/r)}
     <\frac{1}{30}\cdot 2^{-\alpha r\log(d/r)},
\end{align*}
where the first inequality in the second line holds because $r\geq 4$, the next inequality holds since $r\leq 2^{-16}d$ and, for the last inequality, we set $\alpha=\frac{1}{32}$ and use both bounds on $r$. By a union bound over the $q=2^{-\alpha r\log(d/r)}$ queries, 
$\Pr_{\bd, S}[B_{\cT_{na}}|\overline{B_1}]=\Pr_{\bdhat}[B_{\cT_{na}}|\overline{B_1}]
<\frac{1}{30}$.
Using the bound from \Cref{lem: B_1 bound} on the probability of $\overline{B_1}$, we obtain
\begin{align*}
    \Pr_{\bd,S}[B_{\cT_{na}}]
    \leq \Pr_{\bd,S}[B_{\cT_{na}}|\overline{B_1}]+\Pr_{\bd,S}[\overline{B_1}]\
    \leq \frac{2}{30},
\end{align*}
completing the proof of \Cref{lem: B_1 bound}.
\end{proof}

\subsubsection{Proof of Distribution-Free Testing Lower Bound}\label{sec:lowerbound proof}

Before proving \Cref{thm: distribution-free lower bound}, we argue that conditioned on $\overline{B_1\cup B_{\cT}}$, the distribution of samples and query answers seen by $\cT$ is the same whether $(f,U)\sim\pd$ or $(f,U)\sim\nd$.

\begin{definition}[$D$-view]
    \label{def: view}
    For all distributions $D$ over instances of distribution-free testing, and all $t$-sample, $q$-query deterministic algorithms, let $D$-view be the distribution over samples $s_1,\dots,s_t$ and query answers $a_1,\dots, a_q$ seen by the algorithm on input $(f,U)$ when $(f,U)\sim D$.
\end{definition}

\begin{lemma}[Equal conditional distributions]
\label{lem: same view}
     $\pd\text{-view}|_{\overline{B_1\cup B_{\cT}}}=\nd\text{-view}|_{\overline{B_1\cup B_{\cT}}}.$
\end{lemma}
\begin{proof}    
    Conditioned on $\overline{B_1}$ and $\overline{B_{\cT}}$, every query answer $f(x)$ given to $\cT$ is determined by the function $g$ in \Cref{def: T_na}. In particular, every query answer is a deterministic function of the points in $S^+$, the restricted function $f(S^+)$, and (possibly) previous query answers. By \Cref{obs:simulator}, the distribution of $(S^+,f(S^+))$ is the same under both $\pd|_{\overline{B_1}}$ and $\nd|_{\overline{B_1}}$. Hence, the distribution of query answers $f(x_1),...,f(x_q)$ is identical.
    The lemma follows.
\end{proof}

Next, we recall some standard definitions and facts that are useful for proving query lower bounds.

\begin{definition}[Notation for statistical distance]
For two distributions $D_1$ and $D_2$ and a constant $\delta$, let $D_1\approx_{\delta} D_2$ denote that the statistical distance between $D_1$ and $D_2$ is at most $\delta$.
\end{definition}

\begin{fact}[Claim 4 \cite{RS06}]
\label{fact: SD bound}
    Let $E$ be an event that happens with probability at least $1-\delta$ under the distribution $D$ and let $B$ denote the conditional distribution $D|_{E}$. Then $B\approx_{\delta'} D$ where $\delta'=\frac{1}{1-\delta}-1$.
\end{fact}

We use the version of Yao's principle with two distributions from \cite{RS06}.

\begin{fact}[Claim 5 \cite{RS06}]
\label{fact: Yao}
To prove a lower bound $q$ on the worst-case query complexity of a randomized property testing algorithm, it is enough to give two distributions on inputs: $\cP$ on positive instances, and $\cN$ on negative instances, such that 
$\cP\text{-view}\approx_\delta \cN\text{-view}$ for some $\delta<\frac 13$.
\end{fact}

We now complete the proof of \Cref{thm: distribution-free lower bound}, the main theorem on distribution-free testing.

\begin{proof}[Proof of \Cref{thm: distribution-free lower bound}]
We apply \Cref{fact: Yao} (Yao's principle) with $\cP=\pd|_{\overline{B_0}}$ and $\cN=\nd|_{\overline{B_0}}$. By \Cref{lem:dist-to-Lip}, $\pd|_{\overline{B_0}}$ is over positive instances and $\nd|_{\overline{B_0}}$ is over negative instances of distribution-free Lipschitz $\frac{1}{2}$-testing.
Let $\delta_0=\Pr[B_0]$ and $\delta_1=\Pr[B_1\cup B_{\cT}]$. 
Set $\delta_0'=\frac{1}{1-\delta_0}-1$ and $\delta_1'=\frac{1}{1-\delta_1}-1$.
By \Cref{fact: SD bound}, we have the following chain of equivalences:
    \begin{align*}
        \pd\text{-view}|_{\overline{B_0}} 
        \approx_{\delta_0'}\pd\text{-view}
        \approx_{\delta_1'}\pd\text{-view}|_{\overline{B_1\cup B_{\cT}}}
        =\nd\text{-view}|_{\overline{B_1\cup B_{\cT}}}
        \approx_{\delta_1'}\nd\text{-view}
        \approx_{\delta_0'}\nd\text{-view}|_{\overline{B}},
    \end{align*}
        where the equality follows from \Cref{lem: same view}. By Lemmas \ref{lem:dist-to-Lip}, \ref{lem: B_1 bound} and \ref{lem: BcT bound} (that upper bound the probabilities of bad events), for sufficiently large $d$, we have $\delta_0'\leq\frac{1}{27}$, and $\delta_1'\leq\frac{1}{9}$. Hence, $2(\delta_0'+\delta_1')<\frac{1}{3}$. \Cref{thm: distribution-free lower bound} now follows from Yao's principle (as stated in \Cref{fact: Yao}).
\end{proof}

\section{Application to Differential Privacy}\label{sec:privacy-application}

In this section, we show how to use a local Lipschitz filter for bounded-range functions to construct a mechanism (\Cref{thm:filter mechanism}) for privately releasing outputs of bounded-range functions 
even when the client is malicious (i.e., lies about the range or Lipschitz constant of the function). 
Then, we show how the mechanism can be extended to privately release outputs of unbounded-range functions %
(\Cref{thm:binary search mechanism}).

\subsection{Preliminaries on
Differentially Private %
Mechanisms}\label{sec:privacy-prelims}

We start by defining the Laplace mechanism, used in the proofs of Theorems~\ref{thm:filter mechanism} and \ref{thm:binary search mechanism}.
It is based on the Laplace distribution, denoted $\text{Laplace}(\lambda)$, that has
probability density function $f(x %
)=\frac{1}{2\lambda}e^{-|x|/\lambda}$. We use abbreviation $(\epsilon,\delta)$-DP for ``$(\epsilon,\delta)$-differentially private'' (see \Cref{def: privacy}).
\begin{lemma}[Laplace Mechanism \cite{DworkMNS06}]
    \label{lem:laplace}
    Fix $\eps>0$ and $c>1$. Let $f:[n]^d\to\R$ be a $c$-Lipschitz function. Then the %
    mechanism that gets a query 
    $x\in [n]^d$ as input, samples $N\sim\text{Laplace}(\frac c\eps)$, and outputs $L(x)=f(x)+N$, is 
    $(\eps,0)$-DP. Furthermore, for all $\alpha\in (0,1)$, the mechanism satisfies
    $|L(x)-f(x)|\leq \frac c{\eps} + \ln \frac 1 \alpha$ with probability at least $1-\alpha$.
\end{lemma}

In addition to the Laplace mechanism (\Cref{lem:laplace}), the proof of \Cref{thm:binary search mechanism} uses the following well known facts about differentially private algorithms. These can be found in \cite{BookAlgDP}.

\begin{fact}[Composition]
    \label{fact:composition}
    Fix $\eps_1,\eps_2>0$ and $\delta_1,\delta_2\in(0,1)$. Suppose $\cM_1$ and $\cM_2$ are (respectively) $(\eps_1,\delta_1)$-DP and $(\eps_2,\delta_2)$-DP. Then, the mechanism that, on input $x$, outputs $(\cM_1(x),\cM_2(x))$ is $(\eps_1+\eps_2,\delta_1+\delta_2)$-DP. 
\end{fact}

\begin{fact}[Post-processing]
    \label{fact:postprc}
    Fix $\eps>0$ and $\delta\in(0,1)$. Suppose $\cM:D\to R$ is an $(\eps,\delta)$-DP mechanism. If $\cA$ is an algorithm with input space $R$ then the algorithm given by $\cA\circ\cM$ is $(\eps,\delta)$-DP.  
\end{fact}

\begin{fact}
    \label{fact: laplace concentration}
    If $X\sim\text{Laplace}(\lambda)$ then $\Pr[|X|\geq t\lambda]\leq e^{-t}$ for all $t>0$.
\end{fact}
In particular, if $\lambda=\frac{\log r}{\eps}$ and $t=\log(200\log r)$ then $\Pr[|X|\geq\frac{\log(r)\log(200\log r)}{\eps}]\leq \frac{1}{200\log r}$.

\subsection{Mechanism for Bounded-Range Functions}
The filter mechanism can be instantiated with either one of our filters (from \Cref{thm: ell_1 LCA} or from \Cref{thm: ell_0 LCA}), providing slightly different accuracy guarantees. In \Cref{thm:filter mechanism}, we state the guarantees for the mechanism based on the $l_1$-respecting filter. Next, we establish the terminology used in the theorem. Recall that $\ball{x}{R}$ denotes the set $\{y:|x-y|<R\}$. We say a vertex $x\in\zo^d$ is \emph{dangerous} w.r.t.\ $f$ if there exists a vertex $y\in\zo^d$ such that $|f(x)-f(y)|>\dist_G(x,y)$. 
A client that submits a Lipschitz function is called {\em honest}; a client that submits a non-Lipschitz function that is close to Lipschitz is called {\em clumsy} (the distance measure could be $\ell_1$ or $\ell_0$, depending on the filter used). Finally, we assume that sampling from the Laplace distribution requires unit time.

\begin{theorem}[Filter mechanism]
\label{thm:filter mechanism}
For all $\eps>0$ and $\delta\in(0,1)$, there exists an $(\eps,\delta)$-differentially private mechanism $\mathcal{M}$ that, given a query $x\in[n]^d$, lookup
access to a function $f:[n]^d \to [0,r]$, and range diameter $r\in\R$, outputs a value $h(x)\in\R$, and has the following properties.
\begin{itemize}
 \item \textbf{Efficiency: } The lookup and time complexity of $\cM$ are $(d^r \cdot \polylog(n/\delta))^{O(\log r)}$.

\item \textbf{Accuracy for an honest client: }
If $f$ is Lipschitz then for all $x\in[n]^d$ we have $h(x)\sim f(x)+\text{Laplace}(\frac 2\eps)$. %
\item \textbf{Accuracy for a clumsy client: } For all $x\in[n]^d$ such that $\ball{x}{r\log_{3/2}(r)}$ does not contain any dangerous vertices, the ``accuracy for an honest client'' guarantee holds. Moreover, with probability at least $1-2\delta$, the mechanism satisfies $\Exp_{z \sim [n]^d}[|h(z)-f(z)|]\leq2\ell_1(f,\mathcal{L}ip(\hypergrid))+ O(\frac 1 \eps)$.

\end{itemize}
\end{theorem}
We stress that the differential privacy guarantee in \Cref{thm:filter mechanism} holds whether or not the client is honest.

\begin{proof}[Proof of \Cref{thm:filter mechanism}]
Fix $\eps>0$ and $\delta\in(0,1)$. 
Let $\cA$ denote $\textsc{LocalFilter}_1$ (\Cref{alg: ell_1 Filter LCA}) run with iteration parameter $t=\log_{3/2}(r/2)+1$. Recall that by \Cref{thm: ell_1 LCA} instantiated with $\err=1$ and failure probability $\delta$, the algorithm $\cA$ is an $\ell_1$-respecting local $(2,\delta)$-Lipschitz filter with blowup $2$ over the $d$-dimensional hypergrid $\hypergrid$. The ``efficiency'' and ``accuracy for honest client'' guarantees hold for any local Lipschitz filter of the type stated in \Cref{thm: ell_1 LCA}. However, the first guarantee of ``accuracy for a clumsy client'' requires properties specific to the construction of $\textsc{LocalFitler}_1$. 

Let $\cM$ be the following mechanism: \emph{Sample a random seed $\rho$ of length specified in \Cref{thm: ell_1 LCA}, run $\cA(x,\rho)$ to obtain $g_\rho(x)$, sample $N\sim\text{Laplace}(\frac{2}{\eps})$, and output $g_\rho(x)+N$.}

First, we prove that $\cM$ is $(\eps,\delta)$-differentially private.
If the function $g_\rho$ is 2-Lipschitz, then, by \Cref{lem:laplace} instantiated with $c=2$ and privacy parameter $\eps$, the mechanism $\cM$ is $(\eps,0)$-DP. Conditioned on the event that $g_\rho$ is $2$-Lipschitz, we obtain that for all measurable sets $Y \subset \R$,

\[\Pr[\cM(x, \rho)\in Y ~|~ \text{$g_\rho$ is $2$-Lipschitz}]\leq e^{\eps}\Pr[\cM(x', \rho)\in Y ~|~ \text{$g_\rho$ is $2$-Lipschitz}].\]
By \Cref{thm: ell_1 LCA}, $\cA$ fails to output a $2$-Lipschitz function with probability at most $\delta$. By the law of total probability, %
\begin{align*}
\Pr[\cM(x, \rho)\in Y]
&\leq e^{\eps}\Pr[\cM(x', \rho)\in Y ~|~ \text{$g_\rho$ is $2$-Lipschitz}]\Pr[\text{$g_\rho$ is $2$-Lipschitz}] + \delta \\
&\le e^{\eps}\Pr[\cM(x', \rho)\in Y] + \delta.
\end{align*}

The efficiency guarantee follows directly from \Cref{thm: ell_1 LCA}.
The accuracy guarantee for an honest client holds since \Cref{thm: ell_1 LCA} guarantees that if $f$ is Lipschitz, then $\cA(x,\rho) = f(x)$ for all $x$ and $\rho$. 
The average accuracy guarantee for the clumsy client follows from the $\ell_1$-respecting, $2$-blowup guarantee of \Cref{thm: ell_1 LCA} and the fact that $\Exp[|\text{Laplace}(\frac 2\eps)|]\leq O(\frac1\eps)$. To demonstrate that the stronger accuracy guarantee holds when no dangerous vertices are in $\ball{x}{r\log_{3/2}(r)}$, we make the following observation: if no vertex $y\in\ball{x}{r\log_{3/2}(r)}$ is dangerous, then $\cA(x,\rho)=f(x)$. We prove this claim as follows. Recall that $t=\log_{3/2}(r/2)+1$ and, for each $i\in [t]$, let $\cA(x,i,\rho)$ denote the output after iteration $i$ of \Cref{alg: ell_1 Filter LCA}. By construction, $\cA(x,1,\rho)=f(x)$. Suppose no vertex $y\in\ball{x}{r(t-1)}$ is dangerous. Then, since the range of $f$ is $[0,r]$, every dangerous vertex $v$ can, in a single iteration, only create new dangerous vertices in $\ball{v}{r}$. Thus,
in $t-1$ iterations, no dangerous vertices can be introduced in $\ball{x}{r}$. Since $\cA(x,1,\rho)=f(x)$, and in every subsequent iteration $1<i\leq t$, no dangerous vertices are in $\ball{x}{r(i-1)}$, we obtain $\cA(x,t,\rho)=\cA(x,1,\rho)=f(x)$ for each $i\in[t]$. Thus, if no dangerous vertex is in $\ball{x}{r\log_{3/2}(r)}$ then $\cA(x,\rho)=f(x)$, and hence, the ``accuracy for an honest client'' guarantee holds.
\end{proof}

\subsection{Mechanism for Unbounded Range Functions}\label{sec:privacy-examples}

In this section, we use the mechanism for bounded-range functions to construct a mechanism for arbitrary-range functions.

\begin{theorem}[Binary search filter mechanism]
\label{thm:binary search mechanism}
		For all $\eps > 0$ and $\delta \in (0,\frac{1}{200})$, there exists an $(\eps, \delta)$-differentially private mechanism $\cM$ that, given a query $x\in[n]^d$, lookup access to a function $f: [n]^d \to [0,\infty)$ and an optional range %
  parameter $r\in\R$, outputs value $h(x)\in\R$ and has the following properties. 
  
        Let $\kappa=\log\min(r,nd)$, where the optional parameter $r$ is set to $\infty$ by default. %
\begin{itemize}
\item \textbf{Efficiency:} The lookup and time complexity of $\cM$ are $d^{O(\frac1\eps\kappa\log\kappa)}\polylog(\frac n\delta)$.
\item \textbf{Accuracy for an honest client:} If $f$ is Lipschitz then, for all $x\in[n]^d$, we have $h(x)\sim f(x)+\text{Laplace}(\frac{\kappa}{\eps})$ with probability at least $0.99$.
\item \textbf{Accuracy for a clumsy client:} There exists a constant $c>0$ such that for all $x\in[n]^d$, if $f(x)\leq nd$ and $|f(x)-f(y)|\leq |x-y|$ for all $y\in\ball{x}{\frac c\eps\kappa\log\kappa}$, then the ``accuracy for an honest client" guarantee holds. %
\end{itemize}
\end{theorem}

As in \Cref{thm:filter mechanism}, we emphasize that the differential privacy guarantee holds whether or not the client is honest. Note that the accuracy guarantee for an honest client is subsumed by the guaranty for a clumsy client, but we state the former guarantee separately for clarity.

\begin{proof}[Proof of \Cref{thm:binary search mechanism}]
Our private mechanism is presented in \Cref{alg: binary search mechanism}. It uses the following ``clipping" operation to truncate the range of a function.

\begin{definition}[Clipped function]
\label{def:truncation}
For any $f:V \to \R$ and interval $[\ell, u] \subset \R$,  the {\em clipped function} $f[\ell,u]$ is defined by 
\[f[\ell, u](x) = \begin{cases}
    f(x) & f(x) \in [\ell, u]; \\
    \ell & f(x) < \ell; \\
    u & f(x) > u.
\end{cases}\]
\end{definition}

\begin{algorithm}[H]
		\textbf{Input:} Dataset $x \in [n]^d$, lookup access to $f:[n]^d\to[0,\infty)$, range diameter $r \in \R$, adjacency lists access to the hypercube $\hypercube$, $\eps>0$, and $\delta\in(0,1)$\\
\textbf{Subroutines: } Local $(1,\frac{\delta}{\log r})$-Lipschitz filter $\cA$ obtained in \Cref{thm: ell_0 LCA}\\
\textbf{Output:} Noisy value $h(x)$ satisfying the guarantees of \Cref{thm:binary search mechanism}
	\begin{algorithmic}[1]
		\caption{\label{alg: binary search mechanism} Binary search filter mechanism $\cM(x, f,\eps)$}
		\State set $r\gets\min(r,nd)$ and $f\gets f[0,r]$ \Comment{If the client is honest then $f[0,r]=f$.}
		\State set $t \gets r/2$ and $\alpha\gets \frac1\eps\log(r)\log(200\log r)$
		\For{$i=2$ {\bf to} $\lceil \log r \rceil$}
		\State let $h(x) \gets \cA(x, f[t - 2\alpha, t + 2\alpha]) + \mathrm{Laplace}(\frac{\log r}{\eps})$ 
		  \If{$h(x)\in[t-\alpha,t+\alpha]$}
            \State \Return $h(x)$
        \Else
        { $t \gets t + \mathrm{sign}(h(x) - t) \cdot \lceil r/2^i \rceil$}
        \EndIf
		\EndFor
        \State \Return $h(x)$
	\end{algorithmic}
\end{algorithm}

Next, we complete the analysis of the binary search filter mechanism.
    We first argue that $\cM$ (\Cref{alg: binary search mechanism}) is $(\eps,\delta)$-DP. 
    In every iteration of the {\bf for}-loop, $\cM$ uses Laplace mechanism on a function that is 1-Lipschitz
    with probability at least $1-\frac{\delta}{\log r}$. By \Cref{lem:laplace} and an argument similar to the proof of \Cref{thm:filter mechanism}, each iteration is$(\frac{\eps}{\log r}, \frac{\delta}{\log r})$-DP. It follows by Facts \ref{fact:composition} and \ref{fact:postprc} that \Cref{alg: binary search mechanism} is $(\eps,\delta)$-DP. 

    Next, we  prove the accuracy guarantee for the clumsy client. Observe that it subsumes the accuracy guarantee for the honest client. Suppose $f:[n]^d\to[0,r]$ and that $x$ satisfies $|f(x)-f(y)|\leq |x-y|$ for all $y\in\ball{x}{\alpha}$ (the $\alpha$ in line 2 of \Cref{alg: binary search mechanism}). Then, for all intervals $\cI$ of diameter $\alpha$, the point $x$ satisfies $|f[\cI](x)-f[\cI](y)|\leq |x-y|$ for all $y\in[n]^d$. By \Cref{thm: ell_0 LCA}, $\cA(x,f[\cI])=f[\cI](x)$, which is equal to $f(x)$ whenever $f(x)\in \cI$. 
 
 Condition on the event that $\cA$ does not fail and that the Laplace noise added is strictly less than $\alpha$ in every iteration of $\cM$. Then, if $h(x)\in [t-\alpha,t+\alpha]$, we must have $f[t-2\alpha,t+2\alpha](x)\in(t-2\alpha,t+2\alpha)$, and therefore $f[t-2\alpha,t+2\alpha](x)=f(x)$. Next, suppose $h(x)\not\in [t-\alpha,t+\alpha]$. If $f(x)<t$ then $h(x)<t+\alpha$ and thus $h(x)<t-\alpha$. Similarly, if $f(x)>t$ then $h(x)>t-\alpha$ and thus $h(x)>t+\alpha$. It follows that in every iteration $\cM$ either continues the binary search in the correct direction, or halts and outputs $h(x)$ such that $|h(x)-f(x)|\leq \alpha$. By the union bound and the guarantee obtained in \Cref{thm: ell_0 LCA}, the algorithm $\cA$ fails in some iteration of $\cM$ with probability at most $\delta$. Moreover, by \Cref{fact: laplace concentration} and the union bound, the Laplace noise added is at least $\alpha$ in some iteration of $\cM$ with probability at most $\frac{1}{200}$. It follows that for sufficiently small $\delta$, the mechanism $\cM$ outputs $h(x)$ such that $h(x)\sim f(x)+\text{Laplace}(\frac{\log\min(r,nd)}{\eps})$ %
 with probability at least $\frac{99}{100}$.
\end{proof}

\section{Application to Tolerant Testing}\label{sec:tolerant-testing}
In this section, we
give an efficient algorithm for tolerant Lipschitz testing of real-valued functions over the $d$-dimensional hypercube $\hypercube$ and prove \Cref{thm: lipschitz tolerant tester}, which we restate here for convenience.

\begin{restatable}{theorem}{tester}
	\label{thm: lipschitz tolerant tester}
	For all $\eps\in(0,\frac13)$ and all sufficiently large $d\in\N$, there exists an $(\eps,2.01\eps)$-tolerant tester for the Lipschitz property of functions on the hypercube $\hypercube$. The tester has query and time complexity $\frac{1}{\eps^2}d^{O(\sqrt{d\log(d/\eps)})}$.
\end{restatable}

Our tester utilizes the fact that the image of a function which is close to Lipschitz exhibits a strong concentration about its mean on most of the points in the domain. Hence, if a function is close to Lipschitz, it can be truncated to a small interval around its mean without modifying too many points. This truncation guarantees that the the local filter in \Cref{thm: ell_0 LCA} runs in time subexponential in $d$. A key idea in the truncation procedure is that if a function $f$ is $\eps$-close to Lipschitz then either not very many values are truncated, or the truncated function is close to Lipschitz.

To prove \Cref{thm: lipschitz tolerant tester}, we design an algorithm (\Cref{alg: lipschitz tolerant tester}) that, for functions $f:\zo^d\to\R$, accepts if $f$ is $\eps$-close to Lipschitz, rejects if $f$ is $2.01\eps$-far from Lipschitz and fails with probability at most $\frac{45}{100}$. The success probability can then be amplified to at least $\frac23$ by repeating the algorithm $\Theta(1)$ times and taking the majority answer. Before presenting \Cref{alg: lipschitz tolerant tester}, we introduce some additional notation. For all functions $f$ and intervals $\cI\subset\R$, the partial function $f_\cI$ 
is defined by $f_\cI(x)=f(x)$
whenever $f(x)\in \cI$ and $f_\cI(x)=?$ otherwise. Additionally, for all events $E$, let $\mathbf{1}_{E}$ denote the indicator for the event $E$. \Cref{alg: lipschitz tolerant tester} runs the $\ell_0$-filter given by \Cref{alg: ell_0 Filter LCA} with lookup access to a partial function $f_\cI$. Our analysis of \Cref{alg: ell_0 Filter LCA} presented in \Cref{thm: ell_0 LCA} is for total functions.  In \Cref{obs: partial filter}, we extend it to partial functions.

\begin{observation}
    \label{obs: partial filter}
    Let $h:[n]^d\to\R\cup\{?\}$ be a partial function with $\ell_0$-distance to the nearest Lipschitz partial function equal to  $\eps_h$. Let $\cA$ denote \Cref{alg: ell_0 Filter LCA}. Then, for all $\delta\in(0,1)$, the algorithm $\cA^h$ provides query access to a Lipschitz partial function $g$ such that $g(x)=?$ if and only if $h(x)=?$, and $\|g-h\|_0\leq 2\eps_h$. The runtime and failure probability guarantees are as in \Cref{thm: ell_0 LCA}. 
\end{observation}
\begin{proof}
    Consider $x\in h^{-1}(?)$. By \Cref{def: Violation Score}, $VS_h(x,y)=0$ for all $y\in[n]^d$, and hence the vertex $x$ is 
    not incident to any edge of the violation graph of $h$. By construction, $\cA^h(x)=h(x)=?$. Next, consider the induced subgraph $G$ of $\hypercube$ with vertex set $V=\{x:h(x)\neq ?\}$. By \Cref{claim:l0 distance}, the size of the minimum vertex cover of the violation graph of $h$ is at most $|V|\eps_h$. It follows that $\cA^h$ provides query access to a Lipschitz partial function $g$ such that $\|g-h\|_0\leq \eps_h$. The runtime and failure probability are the same as for total functions by definition of the algorithm.   
\end{proof}

\begin{algorithm}[H]
\textbf{Input:} Query access to $f:\zo^d\to\R$, adjacency lists access to $\hypercube$, and $\eps\in(0,\frac13)$\\
\textbf{Subroutines: } Local $(1,\frac{1}{100})$-Lipschitz filter $\cA$ given by \Cref{alg: ell_0 Filter LCA}\\
\textbf{Output:} accept or reject
	\begin{algorithmic}[1]
		\caption{\label{alg: lipschitz tolerant tester} Tolerant Lipschitz tester $\cT(f,\eps)$}
        \State sample a point $p\sim\zo^d$ uniformly at random and  a random seed $\rho$ of length specified in \Cref{thm: ell_0 LCA}
        \State set $t\gets 2\sqrt{d\log(d/\eps)}$ and  $\cI\gets [f(p)-t,f(p)+t]$
		\State sample a set $S$ of $(\frac{1500}{\eps})^2$ points uniformly and independently from $\zo^d$
		\For{{\bf all} $x_i\in S$}
			\If{$f_\cI(x_i)=?$}{ set $y_i\gets ?$}
                \Else
                { set $y_i\gets\cA(x_i, \rho)$\text{, where $\cA$ is run with lookup access to $f_\cI$ and adjacency lists access to $\hypercube$}}
                \EndIf
		\EndFor
		\If{$\frac{1}{|S|}\sum_{x_i\in S}\mathbf{1}_{f(x_i)\neq y_i}<2.005\eps$}\textbf{ accept}
		\Else\textbf{ reject}
		\EndIf		
	\end{algorithmic}
\end{algorithm}

We use  McDiarmid's inequality \cite{McDiarmid89}, stated here for the special case of  the $\{0,1\}^d$ domain. 

\begin{fact}[McDiarmid's Inequality\cite{McDiarmid89}]
\label{fact: mcdiarmid}
    Fix $d\geq 2$ and let $g:\zo^d\to\R$ be a Lipschitz 
    function w.r.t.\ $\hypercube$. Let $\mu_g=\Exp_{x\sim\zo^d}[g(x)]$. Then, for all $\gamma\in(0,1)$,
    $$\Pr_{x\sim\zo^d}[|g(x)-\mu_g|\geq \sqrt{d\log(d/\gamma)}]\leq \frac{\gamma}{d}.$$
\end{fact}

Next, we introduced a definition which, for each function $f$, attributes some part of its $\ell_0$-distance to Lipschitz to a particular interval $\cI$ in the range of $f$.

\begin{definition}
    \label{def:eps[I]}
    Let $f:\zo^d\to\R$ and $C$ be a minimum vertex cover of the violation graph of $f$. (If there are multiple vertex covers, use any rule to pick a canonical one.) For an interval $\cI\subset\R$ define $\eps[\cI]$ as $|\{x\in C:f(x)\in \cI\}|/2^{d}$. 
\end{definition}

For a function $f$, let $\eps_f$ denote 
the $\ell_0$-distance from $f$ to Lipschitz. 
Then $\eps_f=\eps[\cI]+\eps[\overline \cI]$ for all intervals $\cI$. Moreover, since $f_\cI$ is a partial function defined only on points $x$ such that $f(x)\in \cI$, the distance from $f_\cI$ to the nearest Lipschitz partial function is at most $\eps[\cI]$. Using \Cref{def:eps[I]}, we argue that if $\eps_f\leq\eps$, 
then with high probability the interval $\cI$ chosen in \Cref{alg: lipschitz tolerant tester} satisfies $\|f_\cI-f\|_0\leq \eps[\overline \cI]+\frac\eps d$. Since $\eps[\cI]+\eps[\overline \cI]\leq\eps$, \Cref{lem: lipschitz concentration} implies that if $\|f_\cI-f\|_0$ is very close to $\eps$, then the distance of $f_\cI$ to Lipschitz, which is at most $\eps[\cI]$, must be small. Leveraging this fact, we can approximate the $\ell_0$-distance from $f$ to Lipschitz using the distance from $f$ to $f_\cI$ and the distance from $f_\cI$ to Lipschitz. 

\begin{lemma}
    \label{lem: lipschitz concentration}
    Fix $\eps\in(0,\frac13)$ and $d\geq 4$. Let $f:\zo^d\to\R$ be $\eps$-close to Lipschitz over $\hypercube$. Choose $p\in\zo^d$ uniformly at random. Set $t\gets 2\sqrt{d\log(d/\eps)}$ and  $\cI\gets [f(p)-t,f(p)+t]$. Then  $\|f-f_\cI\|_0> \eps[\overline \cI]+\frac\eps d$ with probability at most $\frac{5}{12}$. 
\end{lemma}

\begin{proof}
    Let $C$ be a minimum vertex cover of the violation graph  $B_{0,f}$ (see \Cref{def: Violation Graph}). Let $g$ be a Lipschitz function obtained by extending $f$ from $\zo^d\setminus C$ to $\zo^d$ (such an extension exists by \Cref{claim:standard-extension}).  By \Cref{fact: mcdiarmid}, $\Pr_{x\sim\zo^d}[|g(x)-\mu_g|\geq \frac t2]\leq\frac{\eps}{d}$.
     Notice that if $|f(p)-\mu_g|\leq\frac{t}{2}$ then $[\mu_g-\frac{t}{2},\mu_g+\frac{t}{2}]\subset \cI$. Conditioned on this occurring,
    $$\Pr_x[f(x)\not\in \cI]\leq\eps[\overline \cI]+\Pr_x[g(x)\not\in \cI]\leq \eps[\overline \cI] +\frac\eps d.$$
    Since $\Pr_{x\sim\zo^d}[x\in C]\leq\eps$, we have $|f(p)-\mu_g|\leq \frac{t}{2}$ with probability at most $\eps+\frac{\eps}{d}\leq\frac{5}{12}$. 
\end{proof}

Next, we argue that, after boosting the success probability via standard amplification techniques, we obtain a $(\eps,2.01\eps)$-tolerant Lipschitz tester.

\begin{proof}[Proof of \Cref{thm: lipschitz tolerant tester}]
Fix $\eps\in(0,\frac13)$ and let $f:\zo^d\to\R$ and let $\cT$ denote \Cref{alg: lipschitz tolerant tester}. Define the following events: Let $E_1$ be the event that local filter $\cA$  fails. Set $\omega=\Pr_x[f(x)\neq\cA(x,\rho)]$ and $\hat{\omega}=\frac{1}{|S|}\sum_{x_i\in S}\mathbf{1}_{f(x_i)\neq y_i}$, and let $E_2$ be the event that $|\omega-\hat{\omega}|\geq\frac{\eps}{300}$.

Suppose $f$ is $\eps$-close to Lipschitz, and let $E_3$ be the event that the interval $\cI$ chosen in $\cT$ satisfies $\|f-f_\cI\|>\eps[\overline \cI]+\frac{\eps}{d}$. Condition on the event that none of $E_1,E_2$, and $E_3$ occur. Then $f_\cI$ is at distance at most $\eps[\cI]$ from some Lipschitz partial function and, by \Cref{obs: partial filter}, $\cA$ provides query access to a Lipschitz partial function $g$ such that $\|f-g\|_0\leq 2\eps[\cI]$. Using the fact that $\eps[\cI]+\eps[\overline \cI]\leq\eps$ we obtain 
$$\hat\omega\leq2\eps[\cI]+\eps[\overline \cI]+\frac\eps d+\frac{\eps}{300}<2.005\eps$$ 
for sufficiently large $d$. Hence $\cT$ accepts.

Now consider the case that $f$ is $2.01\eps$-far from Lipschitz and suppose neither of the events $E_1$ and $E_2$ occur. Since $\cA$ provides query access to a Lipschitz function, and the nearest Lipschitz function is at distance at least $2.01\eps$, we must have $\hat\omega\geq2.01\eps-\frac{\eps}{300}>2.005\eps$. Hence $\cT$ rejects.

Next, we show that the events $E_1,E_2$ and $E_3$ all occur with small probability. %
By \Cref{obs: partial filter}, $\Pr[E_1]\leq\frac{1}{100}$. To bound $\Pr[E_2]$, notice that $\Exp[\hat{w}]=w$ and $\Var[\hat{w}]\leq\frac{1}{4|S|}$. By Chebyshev's inequality and our choice of $|S|$ we have, $\Pr_S[|w-\hat{w}|\geq\frac{\eps}{300}]\leq\frac{300^2}{4|S|\eps^2}\leq\frac{1}{100}$. Moreover, if $f$ is $\eps$-close to Lipschitz, then by \Cref{lem: lipschitz concentration}, $\Pr[E_3]\leq\frac{5}{12}$. Thus, the failure probability of $\cT$ can be bounded above by $\Pr[E_1\cup E_2\cup E_3]\leq\frac{2}{100}+\frac{5}{12}\leq\frac{45}{100}$. The success probability can the be boosted to $\frac{2}{3}$ by running the algorithm $O(1)$ times and taking the majority answer. %
Finally, we bound the query and time complexity of tester by bounding the query and time complexity of $\cT$. The algorithm $\cT$ runs the local filter $\cA$  with lookup access to $f_\cI$, a function with range $\cI$ of diameter $O(\sqrt{d\log(d/\eps)})$, and sets $\cA$'s failure probability to $\delta=\frac{1}{100}$. Consequently, \Cref{obs: partial filter} implies $\cT$ has query and time complexity bound of $\frac{1}{\eps^2}d^{O(\sqrt{d\log(d/\eps)})}$.
\end{proof}

\subsubsection*{Acknowledgement}
We thank Adam Smith for comments on the initial version of this paper.

\bibliographystyle{alpha}
\bibliography{references}

\end{document}